%% file: main.tex
\documentclass[runningheads]{llncs}

\input{sections/macro.tex}

\usepackage{soul}
\usepackage{autobreak}
\usepackage{pifont}
\usepackage{rotfloat}
\usepackage{ifsym}
\usepackage{bbding}
\usepackage{mathtools}
\usepackage{lineno}
\usepackage{algpseudocode}
\usepackage{makecell}
\usepackage{threeparttable}
\usepackage{booktabs}
\usepackage{multirow}
\usepackage{graphicx}
\usepackage{xcolor}
\usepackage{colortbl}

\theoremstyle{definition}

\begin{document}
\title{Normal Alignment: Improved Cryptanalytic Sign Recovery on Hard-Label Networks}
\titlerunning{Normal Alignment for Sign Recovery on Hard-Label Networks}
\author{
Shi Tang\inst{1} \and
Zirui Chen\inst{2} \and
Yongjia Su\inst{1} \and
Zhengchao Gao\inst{1} \and\\
Lingyue Qin\inst{2} \and
Xiaoyang Dong\inst{2}
}
\authorrunning{S. Tang et al.}

\institute{
Shandong University, Jinan, P. R. China\\
\email{\{shi.tang,yongjia.su,chao\_qwq\}@mail.sdu.edu.cn}
\and
Tsinghua University, Beijing, P. R. China\\
\email{chenzr25@mails.tsinghua.edu.cn}\\
\email{\{qinly,xiaoyangdong\}@tsinghua.edu.cn}
}

\maketitle

\begin{abstract}
At EUROCRYPT 2025, Carlini {\em et al.} proposed a breakthrough in the cryptanalytic extraction on hard‑label (S1) deep neural networks (DNNs), demonstrating polynomial-time signature and sign recovery. However,  Carlini {\em et al.}'s sign‑recovery method ({which we call \em Future Toggle}) suffers  only a marginal advantage over random guessing, producing high‑confidence wrong sign predictions in deeper layers. Such errors trigger expensive exponential‑time enumeration. 

This work presents {\em Normal Alignment}, a novel statistical sign‑recovery approach for S1 DNNs. Drawing on the expected length difference between projected normals of adjacent decision facets at dual points, our method infers neuron signs via normal‑signature alignment. It delivers higher voting accuracy and pushes erroneous predictions to low‑confidence ranks, which further enables a more efficient combined method, {\em eSOE + Alignment}, by combining {\em Normal Alignment} with the hard‑label {SOE} extension. This combined strategy removes heavy enumeration overhead and realizes exact polynomial‑time full sign recovery. 

Experiments demonstrate the effectiveness of our method,  especially for deep layers. For example, with our method, the signs for CIFAR-10 (architecture 192-64$\times$8-10) and MNIST (architecture 64-96$\times$3-32-10) models can be fully recovered in polynomial time; in contrast, Carlini {\em et al.}'s sign‑recovery method would require exponential‑time enumerations involving $2^{52}$ or $2^{82}$ guesses of the signs, respectively.



\end{abstract}


\input{sections/Introduction.tex}

\input{sections/PreliminariesCompact.tex}
\input{sections/limitations.tex}

\input{sections/hardlabel.tex}
\input{sections/combineSOE.tex}
\input{sections/experiment.tex}
\section{Conclusion}\label{sect:conclusion}
This paper proposes {\em Normal Alignment}, an improved statistical sign‑recovery method for hard‑label ReLU network extraction to address the limitations of {\em Future Toggle}. Using projected decision‑facet normals at dual points, it achieves higher voting accuracy and pushes errors to low‑confidence ranks. Combined with extended hard‑label {\em SOE}, {\em eSOE+Alignment} achieves exact polynomial‑time full sign recovery without exponential enumeration. Evaluations on CIFAR‑10 and MNIST confirm its superiority.  

\bibliographystyle{splncs04}
\bibliography{bib}

\newpage
\appendix
\renewcommand{\theHsection}{Appendix.\arabic{section}}
\section*{\centering \textsf{Supplementary Material}}

\section{Supporting Experimental Results}
\label{supp:supporting-experimental-results}
\vspace{-0.6em}

\input{sections/appendixExperimentalResults}

\end{document}

%% file: sections/macro.tex
\usepackage{makeidx}
\usepackage{multirow}
\usepackage{graphicx}
\usepackage{tabularx}
\usepackage{array}
\usepackage{amsmath}
\usepackage{amssymb}
\usepackage{amsfonts}
\usepackage{url}
\usepackage[figuresright]{rotating}
\usepackage{arydshln}
\usepackage[vlined,boxed,commentsnumbered,ruled,linesnumbered]{algorithm2e}
\usepackage{caption}
\usepackage{subcaption}
\usepackage{lscape}
\usepackage{bm}
\usepackage{braket} 

\usepackage{amsthm}
\usepackage{float}

\usepackage{booktabs} 

\usepackage[colorlinks, citecolor=blue]{hyperref} 

\newcommand{\bx}{\boldsymbol{x}}
\newcommand{\by}{\boldsymbol{y}}
\newcommand{\bb}{\boldsymbol{b}}
\newcommand{\bu}{\boldsymbol{u}}

\newcommand{\bn}{\boldsymbol{n}}
\newcommand{\bg}{\boldsymbol{g}}
\newcommand{\bbeta}{\boldsymbol{\beta}}
\newcommand{\bgamma}{\boldsymbol{\gamma}}
\newcommand{\bdelta}{\boldsymbol{\delta}}

\newcommand{\bA}{\boldsymbol{A}}
\newcommand{\bS}{\boldsymbol{S}}
\newcommand{\bI}{\boldsymbol{I}}
\newcommand{\bF}{\boldsymbol{F}}
\newcommand{\bG}{\boldsymbol{G}}
\newcommand{\bP}{\boldsymbol{P}}
\newcommand{\bQ}{\boldsymbol{Q}}

\providecommand{\tblNA}{\textsc{n/a}}

\usepackage{listings}
\usepackage{color}

\usepackage{pgfplots, pgfplotstable}
\usepgfplotslibrary{fillbetween}
\usetikzlibrary{patterns}

\definecolor{oursblue}{rgb}{0.125,0.306,0.439}

\newcolumntype{C}{>{\centering\arraybackslash}X}

\definecolor{polyshade}{RGB}{190,232,200}
\definecolor{expshade}{RGB}{255,226,150}

\newcommand{\polycell}{%
  \begingroup
  \setlength{\fboxsep}{3.5pt}%
  \colorbox{polyshade}{%
    \makebox[4em][c]{$\mathit{poly}$}%
  }%
  \endgroup
}

\newcommand{\expcell}{%
  \begingroup
  \setlength{\fboxsep}{3.5pt}%
  \colorbox{expshade}{%
    \makebox[4em][c]{$\mathit{exp}$}%
  }%
  \endgroup
}

%% file: sections/Introduction.tex

\section{Introduction}
\label{sec:introduction}


Deep neural networks (DNNs) are widely used in computer vision \cite{Alex2012}, natural-language processing~\cite{DBLP:conf/nips/VaswaniSPUJGKP17}, and medical diagnosis~\cite{DBLP:journals/nature/EstevaKNKSBT17}, etc. Training high-performing DNNs often requires large amounts of data, computation and engineering efforts, making trained models valuable intellectual assets \cite{oliynyk2023know}. 
Model extraction is a long‑studied attack in which an adversary uses input-output queries of the victim DNN (or other side-channel information \cite{batina2019csi}) to extract  its parameters (weights and biases). Early  works explored network reconstruction \cite{fefferman1994reconstructing} and query‑based model stealing \cite{lowd2005adversarial,tramer2016stealing}. At CRYPTO~2020, 
Carlini, Jagielski, and Mironov introduced a cryptanalytic approach to model extraction \cite{DBLP:conf/crypto/CarliniJM20}, which exploits the piecewise-affine structure of a ReLU network and recovers its parameters layer by layer from raw-output queries. Layer-wise extraction proceeds in two stages: First, the \emph{signature recovery} identifies the unsigned weights and biases of each neuron by using the high-order differential at a so-called {\em critical point}, where one ReLU input of the DNN is exactly zero; Second, the \emph{sign recovery} determines the sign. Their method enjoys a polynomial-time \emph{signature recovery} phase, but suffers from an exponential-time \emph{sign recovery} phase by a brute-force guessing method. 
At EUROCRYPT~2024,  Canales-Mart\'inez {\em et al.} \cite{DBLP:conf/eurocrypt/CanalesMartinezCHRSS24}  developed polynomial-time \emph{sign recovery} algorithms in raw-output  setting, {\em i.e.}, the {\em Neuron Wiggle} and {\em SOE} methods. Later work examined the practical limitations of these layer-wise attacks.
Foerster {\em et al.}~\cite{DBLP:conf/nips/FoersterMSH24} found that increasing the number of critical points does not necessarily improve {\em Neuron Wiggle} recovery for difficult neurons. Liu et al.~\cite{DBLP:conf/eurocrypt/LiuSELBP26} addressed rank-deficient signature systems and the misattribution of critical points from deeper layers, thereby extending practical extraction from three hidden layers to eight. 
In parallel, cryptanalytic extraction has expanded along two dimensions: covering various activation functions \cite{chen2025delving,qi2026various,DBLP:conf/crypto/AsselineauDFM26}  and different network architectures \cite{wei2026rnn,cnn_average_pooling,liu2026model-cnn,DBLP:conf/crypto/ChenTGSQD26}. 

\subsubsection{Hard-label extraction.}
The S1 hard-label setting returns only the predicted class labels ({\em e.g.}, ``dog'' or ``car'') and hides the logits. At ASIACRYPT~2024, Yi Chen {\em et al.} initiated the cryptanalytic extraction in the S1 setting~\cite{DBLP:conf/asiacrypt/ChenDGSWW24}, but it requires exponential execution time.  
At EUROCRYPT~2025, Carlini {\em et al.} \cite{DBLP:conf/eurocrypt/CarliniCHRS25} gave the first polynomial-query, polynomial-time hard-label extraction. 
The attack collects and clusters \emph{dual points}, which are critical and also on a visible class-decision boundary, to recover the signatures. 
At CRYPTO 2026, Ito, Miura, and Todo~\cite{DBLP:conf/crypto/ItoMT26} identified a limitation of Carlini {\em et al.}'s attack \cite{DBLP:conf/eurocrypt/CarliniCHRS25}: for nearly always-active neurons, the state switches needed for parameter recovery can become exponentially difficult to observe. They proposed cross-layer extraction to address this failure mode. 
In 2026, Zirui Chen {\em et al.}~\cite{cryptoeprint:2026/1164}  proposed the {\em Approximate Signature Vector} (ASV) method to reduce the cost of clustering dual points, and hence improve Carlini {\em et al.}'s signature recovery phase \cite{DBLP:conf/eurocrypt/CarliniCHRS25}. 
\subsubsection{Existing Sign Recovery Methods in S1 and Their Limitations.} There are only two existing S1 {\em sign recovery} methods: Carlini {\em et al.}'s statistical {\em Future Toggle} method   \cite{DBLP:conf/eurocrypt/CarliniCHRS25}, and Canales-Mart\'inez and Santos's deterministic hard-label {\em SOE} method \cite{DBLP:conf/latincrypt/CanalesMartinezS25}. The hard-label {\em SOE} typically recovers only the first hidden layer, unless the network is sufficiently contractive to permit the recovery of deeper layers. 
 The main limitation of  \textit{Future Toggle} is its   
     {\em weak advantage over random guessing}.  
    In our practical experiments, {\em Future toggle} suffers small advantage than random guessing, leading to the low vote accuracy and the low vote confidence. For example, on the evaluated CIFAR-10 model, Table~\ref{tab:whitebox-future-toggle-cifar10} reports vote accuracies of $52\%$--$57\%$, while Figure~\ref{fig:confidence-cifar-l7-toggle} shows that the confidence of most neurons is below $60\%$.

   The weak advantage may lead to many incorrectly recovered signs in deep layers, which results in an exponential time complexity to recover the full signs, shifting the entire hard-label extraction  from polynomial to exponential time complexity. Specifically, the 
   {\em Future Toggle} may produce errors with relatively high confidence. As shown in Figure~\ref{fig:confidence-cifar-l7-toggle}, 9 of the $64$ signs on layer 7 are incorrectly recovered, including one with a confidence of $55.58\%$ that ranks $20$th among all neurons by confidence. This makes exact recovery of the entire layer difficult under two existing approaches. 
   \begin{itemize}
       \item {\bf Enumeration Infeasible.} To recover all signs, we follow the approach of Foerster et al. \cite{DBLP:conf/nips/FoersterMSH24}: we enumerate and guess the low-confidence sign assignments (which may contain errors) until all erroneous signs are covered, and verify the assignments by executing the next-layer signature recovery algorithm.As shown in Figure~\ref{fig:confidence-cifar-l7-toggle}, since the erroneous sign ranks $20$-th in confidence, the enumeration must cover all $45$ signs from rank $20$ to rank $64$. This requires testing \(2^{45}\) possible assignments.
       \item  {\bf SOE Failure.} In 2026, Liu {\em et al.}~\cite{DBLP:conf/eurocrypt/LiuSELBP26} combined statistical predictions ({\em Neuron Wiggle}) with {\em SOE} in raw-output setting by eliminating unknowns associated with neurons predicted to be inactive with high confidence. Hence, when an actually active neuron is incorrectly predicted to be inactive with high confidence, its nonzero unknown is likely to be eliminated, making the reduced {\em SOE} incorrect and causing the combined strategy to fail.  
   \end{itemize}
  

    \begin{figure}[!t]
    \centering

    \begin{subfigure}[t]{0.49\linewidth}
        \centering
        \includegraphics[width=\linewidth]
        {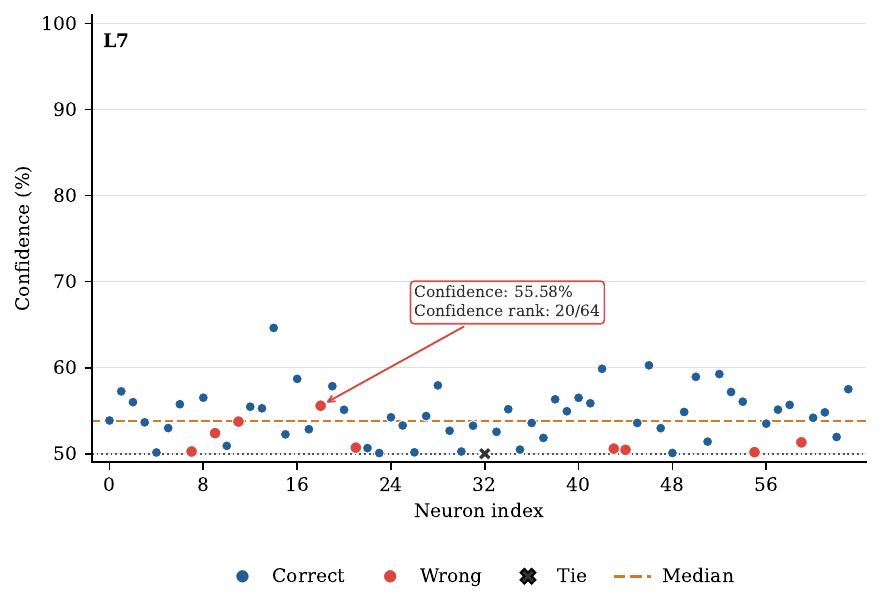}
        \caption{{\em Future Toggle} ($n_{\mathrm{attempt}}=1000$).}
        \label{fig:confidence-cifar-l7-toggle}
    \end{subfigure}
\hfill
    \begin{subfigure}[t]{0.49\linewidth}
        \centering
        \includegraphics[width=\linewidth]
        {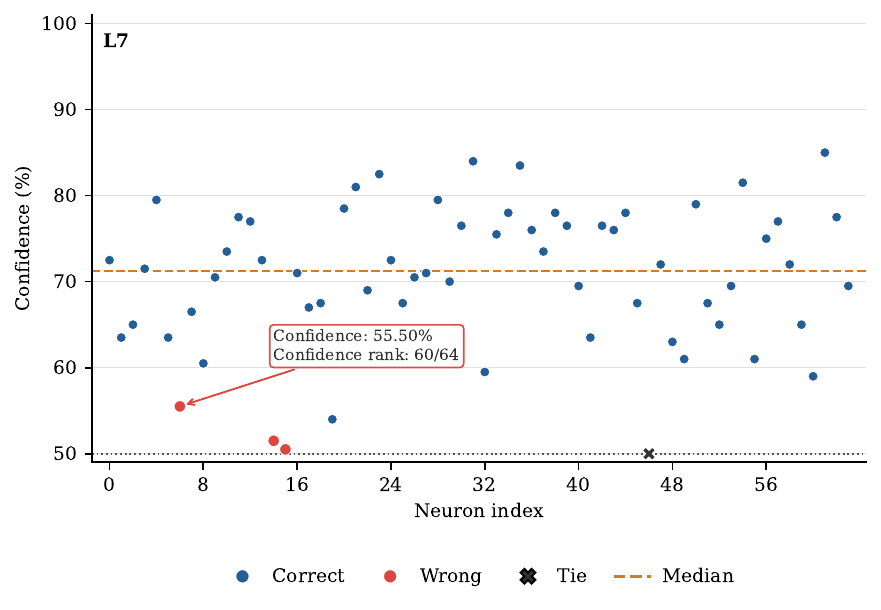}
        \caption{{\em Normal Alignment} ($n_{\mathrm{attempt}}=200$).}
        \label{fig:confidence-cifar-l7-alignment}
    \end{subfigure}
    \caption{vote confidence distributions of {\em Future Toggle} and {\em Normal Alignment} on the 7th hidden layer of the CIFAR-10 model with architecture
    \mbox{$192$-$64{\times}8$-$10$}.}
    \label{fig:confidence-cifar-l7-comparison}
    \end{figure}

\subsubsection{Our Contributions.}  
This paper introduces a new sign recovery method in S1 setting, called {\em Normal Alignment}, which uses the normals of the  two decision facets adjacent to a dual point to directly determine the target sign, without repeatedly walking along the decision boundary to search for neuron toggles in future layers  ({\em i.e.}, Carlini {\em et al.}'s {\em Future Toggle} method \cite{DBLP:conf/eurocrypt/CarliniCHRS25}). Our method is based on the following statistical intuition (which is also formally proved): at a dual point, the decision-facet normal on the active side has a larger expected length than the normal on the inactive side, because it explicitly includes the target neuron's weight contribution. 
The advantages of our methods are summarized below: 
    \begin{enumerate}
    \item {\bf Larger Statistical Advantage.}
    As shown in Table~\ref{tab:mnist-layer-comparison}, with the same budget of $n_{\mathrm{attempt}}=200$ dual points, {\em Normal Alignment} achieves an overall vote accuracy of $68.74\%$ across all hidden layers, compared with $60.89\%$ for {\em Future Toggle}. In deep hidden layers in Table \ref{tab:cifar-layer-comparison}, our method still maintains a vote accuracy of about $70\%$, while the vote accuracy  of {\em Future Toggle}  drops to around $53\%$ in layers 3-7.  
   

    \item {\bf Fewer Queries Needed.}
    {\em Normal Alignment} neither walks along the decision boundary, nor discards dual points due to non-future neuron toggles -- as occurs with  {\em Future Toggle}. When the two adjacent decision-facet normals are successfully recovered, a dual point can produce a vote. Consequently, the dual point utilization reaches $100\%$, as shown in Table~\ref{tab:cifar-layer-comparison}. By contrast,  only $2.26\%$ of the dual points produce a vote in Layer 8 by the {\em Future Toggle}. Furthermore, the weak advantage of the {\em Future Toggle} naturally requires more votes to improve accuracy, and hence  usually needs more queries and time than {\em Normal Alignment} 
   as shown in Table~\ref{tab:intro-combined-sign-recovery}.  

    \item {\bf Low Confidence for Incorrect Signs.}
    As shown in Figure~\ref{fig:confidence-cifar-l7-comparison}, For {\em Normal Alignment}, the highest-confidence error has a confidence of $55.50\%$ and ranks $60$th among the $64$ neurons, whereas {\em Future Toggle}'s highest-confidence error has a
    similar confidence of $55.58\%$ but ranks as high as $20$th. Therefore, to recover the full signs by enumeration \cite{DBLP:conf/nips/FoersterMSH24}, our method must test $2^5$ sign assignments, while {\em Future Toggle} tests $2^{45}$ sign assignments.
    We also test more models with different layers  in Table \ref{tab:highest-error-confidence-rank} in {\sf Supp.}~\ref{supp:supporting-experimental-results} .  Across all hidden layers of each model, our highest-ranked errors occur on CIFAR-10 L6 and MNIST L2 at ranks 50/64 and 68/96, respectively, yielding enumeration complexities of $2^{15}$ and $2^{29}$. By contrast, at the same budget of $n_{\mathrm{attempt}}=200$, the highest-ranked errors of {\em Future Toggle} require enumeration complexities of $2^{64}$ and $2^{93}$ on the CIFAR-10 and MNIST models, respectively.
    
    \item {\bf Feasible Combination of the Hard-label {\em SOE} and {{\em Normal Alignment}: {\em eSOE + Alignment}.}} Because the highest-confidence error ranks very low in {\em Normal Alignment}, a similar combination of a statistical method and a deterministic method by Liu {\em et al.} ~\cite{DBLP:conf/eurocrypt/LiuSELBP26} in S5 setting works in S1 setting, {\em i.e.}, combining hard-label {\em SOE} and {\em Normal Alignment}. 
    The {\em Normal Alignment} identifies the high-confidence (highly ranked) inactive neurons  and eliminates the corresponding zero equations in SOE. By contrast,  {\em Future Toggle} is susceptible to high-confidence errors ({\em i.e.}, errors with high confidence rank); for instance, an active neuron might be predicted as inactive with high confidence, causing a nonzero equation to be erroneously discarded and thus causing the {\em SOE} method to fail. Similarly to Liu {\em et al.} ~\cite{DBLP:conf/eurocrypt/LiuSELBP26}, to further increase the rank of {\em SOE}, we select several transition points sharing the same activation states in the target and future layers, forming a stacked coefficient matrix. Besides, we also introduce an orthogonal projection matrix to eliminate the unknown normal length at each point. Therefore, we call the resulting method  hard-label {\em SOE} extension and {\em Normal Alignment} ({\em eSOE + Alignment}). It helps eliminate the enumeration complexities -- specifically the $2^{15}$ time  for layer L6 of the CIFAR-10 network (Table \ref{tab:cifar-layer-comparison}) and the $2^{29}$ time for layer L2 of the MNIST network (Table \ref{tab:mnist-layer-comparison}).


    \end{enumerate}
    \subsubsection{Experiments.} As shown in Table \ref{tab:intro-combined-sign-recovery}, the sign recovery methods are evaluated on CIFAR-10 and MNIST networks. We follow the same assumption as  \cite{DBLP:conf/eurocrypt/CanalesMartinezCHRSS24}:  when targeting the layer $k$, the preceding layers ($<k$) and the unsigned signatures of layer $k$ are known. 
    Tables \ref{tab:cifar-layer-comparison} and \ref{tab:mnist-layer-comparison} in {\sf Supp.} \ref{supp:supporting-experimental-results} report the full results of our experiments, comparing {\em Normal Alignment}, {\em eSOE + Alignment}, {\em Future Toggle},  the combination of the hard-label {\em SOE} extension and {\em Future Toggle} ({\em eSOE + Toggle}) for a fair comparison though {\em eSOE + Toggle} does not reduce the overall enumeration complexities as shown in Table~\ref{tab:intro-combined-sign-recovery}.
   Our {\em eSOE+Alignment} correctly recovers all $512/512$ signs on CIFAR-10  and all $320/320$ signs on MNIST, thereby achieving exact sign recovery in polynomial time in Table~\ref{tab:intro-combined-sign-recovery}. In contrast, {\em eSOE+Toggle} recovers only $469/512$ and $296/320$ signs, respectively. Specifically,  
  for the CIFAR-10 model with $n_{\rm attempt}=1000$ in Table \ref{tab:cifar-layer-comparison}, {\em eSOE} + {\em Toggle} leaves errors on L5, L7, and L8, whose highest-confidence erroneous signs rank $44$-th, $20$-th, and $13$-th out of 64, respectively. Exact recovery must therefore enumerate the $52$ signs from rank $13$ to rank $64$ in L8, requiring $2^{52}$ sign enumerations. For the MNIST model in Table \ref{tab:mnist-layer-comparison}, the highest-confidence error ranks $15$-th among $96$ neurons on L3, hence requiring $2^{82}$ sign enumerations to recover full signs. 
The source code for all the experiments can be found via
\begin{center}
\url{XXX}
\end{center}

\begin{table}[!t]
\scriptsize
\centering
\caption{Comparison of our {\em eSOE+Alignment} and
{\em eSOE+Toggle}~\cite{DBLP:conf/eurocrypt/CarliniCHRS25}
on the CIFAR-10 and MNIST models.}
\label{tab:intro-combined-sign-recovery}

\setlength{\tabcolsep}{4pt}
\renewcommand{\arraystretch}{1.35}

\resizebox{\textwidth}{!}{%
\begin{threeparttable}

\begin{tabular}{@{}llcccccc@{}}
\toprule

\multicolumn{3}{c}{\textbf{Model and Method}}
& \multicolumn{2}{c}{\textbf{Recovery Complexity}}
& \multicolumn{3}{c}{\textbf{Experimental Results}}
\\

\cmidrule(lr){1-3}
\cmidrule(lr){4-5}
\cmidrule(l){6-8}

\textbf{Architecture}
& \textbf{Method}
& $\bm{n_{\mathrm{attempt}}}$
& \makecell{\textbf{Method}\\\textbf{execution}}
& \makecell{\textbf{Complete}\\\textbf{recovery}}
& \textbf{Correct signs}
& \textbf{Time}
& \textbf{Queries}
\\
\midrule

\multirow{2}{*}{%
    \makecell[l]{
        \textbf{CIFAR-10}\\
        $192$-$64{\times}8$-$10$
    }
}
& eSOE+Alignment
& 200
& $\mathit{poly}$
& \polycell
& $\mathbf{512/512}$
& $2^{10.73}$
& $2^{31.13}$
\\

& eSOE+Toggle~\cite{DBLP:conf/eurocrypt/CarliniCHRS25}
& 1000
& $\mathit{poly}$
& \expcell
& $469/512$
& $2^{13.47} (s)+ 2^{52} (g)$
& $2^{35.49}$
\\

\midrule

\multirow{2}{*}{%
    \makecell[l]{
        \textbf{MNIST}\\
        $64$-$96{\times}3$-$32$-$10$
    }
}
& eSOE+Alignment
& 200
& $\mathit{poly}$
& \polycell
& $\mathbf{320/320}$
& $2^{11.20}$
& $2^{28.71}$
\\

& eSOE+Toggle~\cite{DBLP:conf/eurocrypt/CarliniCHRS25}
& 1000
& $\mathit{poly}$
& \expcell
& $296/320$
& $2^{17.39}(s)+ 2^{82} (g)$
& $2^{30.77}$
\\

\bottomrule
\end{tabular}

\begin{tablenotes}[flushleft]

\item[]
$\boldsymbol{n_{\mathrm{attempt}}}$: denotes the number of attempted dual
points per neuron. As shown in
Fig.~\ref{fig:normal-alignment-dual-point-budget}, we set
$n_{\mathrm{attempt}}=200$ for {\em Normal Alignment}. Because
{\em Future Toggle} provides a weaker statistical advantage, we set
$n_{\mathrm{attempt}}=1000$ for {\em Future Toggle}.

\item[]
\textbf{Method execution}: denotes the time complexity of
{\em eSOE+Alignment} or {\em eSOE+Toggle}. Since hard-label {\em SOE},
{\em Normal Alignment}, and {\em Future Toggle} all run in polynomial
time, both combined methods also run in polynomial time.

\item[]
\textbf{Complete recovery}: denotes the time complexity required to
recover all neuron signs correctly. {\em eSOE+Alignment} recovers all
neuron signs correctly and therefore achieves complete recovery in
polynomial time. In contrast, {\em eSOE+Toggle} leaves some signs
incorrect; guaranteeing complete recovery therefore requires
exponential enumeration of the unresolved sign assignments~\cite{DBLP:conf/nips/FoersterMSH24}.

\item[]
\textbf{Time}: consists of two components. The notation $2^x\, (\mathrm{s})$ denotes the time of {\em eSOE+Alignment/Toggle} in seconds, whereas $2^y\, (\mathrm{g})$ denotes the cost of sign guessing required for complete recovery of all signs in the model. 

\item[]
\textbf{Queries}: In the CIFAR-10 proof-of-concept experiments, the
decision-facet normals are computed directly from model parameters,
and the reported values estimate the corresponding hard-label query
cost; the MNIST entries report the actual hard-label query counts.

\end{tablenotes}
\end{threeparttable}%
}
\end{table}

%% file: sections/PreliminariesCompact.tex

\ifdefined\definitionsstandalone
\section*{Preliminaries}
This document extracts the network and geometric notation, oracle models, and
extraction goal used in the paper.
For a positive integer $m$, we write $[m]=\{1,\ldots,m\}$.
\subsection*{Notations and Definitions}
\else
\section{Preliminaries}
\label{sec:preliminaries}

Unless otherwise specified, the subscript and superscript numbers start from 1. 
\begin{itemize}
\item $[m]$: for a positive integer $m$,
we write $[m]=\{1,\ldots,m\}$, 
    \item $\bA$: matrix, where its $i$-th row is $\bA_{i}$, and its element in $i$-th row and $j$-th column is $\bA_{i,j}$, $i,j\geq 1$,

\item $\bx$: column vector, and its $i$-th element is $\bx_i$, $i\geq 1$, 

\item $\mathcal{F}$: functions,
\item $\mathbb{C}, \mathbb{D}$: space or set,
\item neuron $(k,j)$: the $j$-th neuron in layer $k$. 
\end{itemize}

\subsection{Notations and Definitions}
\label{subsec:notations-definitions}
\fi
\label{subsec:networks-layers-neurons}
\label{subsec:activation-local-linearity}

The DNN is composed of a sequence of functions alternating between linear functions $f^{(k)}: \mathbb{R}^{d^{(k)}}\mapsto \mathbb{R}^{d^{(k+1)}}$ ($k\in [r+1]$), and a nonlinear function $\sigma$ (component-wise ReLU function): 
\begin{equation}
    \mathcal{F}_{\theta}=f^{(r+1)}\circ \sigma\circ f^{(r)}\circ \sigma\circ\cdots f^{(2)}\circ \sigma\circ f^{(1)}, 
\end{equation}
where $f^{(k)}:\mathbb{R}^{d^{(k)}} \rightarrow \mathbb{R}^{d^{(k+1)}}$ is an affine transformation:
        \begin{equation}
            \by^{(k)}=f^{(k)}({\bx}^{(k)})=\bA^{(k)}\bx^{(k)}+\bb^{(k)} \in \mathbb{R}^{d^{(k+1)}},
        \end{equation}
    where $\bx^{(k)}\in\mathbb{R}^{d^{(k)}}$ represents the input vector of layer $k$, and $\bx^{(1)}\in\mathbb{R}^{d^{(1)}}$ is the model input. The weight matrix $\bA^{(k)}\in \mathbb{R}^{d^{(k+1)}\times d^{(k)}}$ and the bias vector $\bb^{(k)}\in \mathbb{R}^{d^{(k+1)}}$ are composed of floating-point numbers, which are the model parameters.  
    
Given input $\bx^{(1)}\in \mathbb{R}^{d^{(1)}}$, the ReLU function $\sigma$ in layer $k$ is also interpreted as the matrix determined by  $\bx^{(1)}$, 
\begin{equation}\label{eqn:I}
    \bI^{(k)}=\operatorname{diag}( \tau^{(k)}_1, \tau^{(k)}_2,\dots, \tau^{(k)}_{d^{(k+1)}}), 
\end{equation}
where $\tau^{(k)}_i=1,~ i\in [d^{(k+1)}]$ when $\by^{(k)}_i\geq 0$, else $\tau^{(k)}_i=0$. Then, $\bx^{(k+1)}=\bI^{(k)}\by^{(k)} \in \mathbb{R}^{d^{(k+1)}}$ is the output of layer $k$. 

\begin{definition}[Linear Neighborhood]
    Given an input  $\bx\in  \mathbb{R}^{d^{(1)}}$, the matrices $\bA^{(k)},~\bb^{(k)},~\bI^{(k)}, ~ k\in [r+1]$ will be all fixed. The linear neighborhood of $\bx$ is defined as the subset $\mathbb{L}_{\bx} \subset \mathbb{R}^{d^{(1)}}$, so that, for all  $\bx^{(1)}\in \mathbb{L}_{\bx}$, the same matrices $\bA^{(k)},~\bb^{(k)},~\bI^{(k)}$ will be  applied to compute the output of the DNN. 
\end{definition}
The DNN has been proved to be a piecewise linear function \cite{DBLP:conf/crypto/CarliniJM20,DBLP:conf/eurocrypt/CanalesMartinezCHRSS24}, {\em i.e.}, for $\bx \in \mathbb{L}_{\bx}$, the model output will change linearly, {\em i.e.},  the DNN is reduced to 
\begin{equation}\small
    \label{eqn:fcnn_linear}
    \begin{array}{lll}
     \mathcal{F}_{\theta}(\bx) &=&  \bA^{(r+1)} \left(\bI^{(r)}\left(\bA^{(r)} \cdots \left(\bI^{(1)}\left(\bA^{(1)}\bx + \bb^{(1)}\right)\right) \cdots + \bb^{(r)}\right) \right)+ \bb^{(r+1)} \\
        &=& \bA^{(r+1)}\bI^{(r)}\bA^{(r)} \cdots \bI^{(2)}\bA^{(2)}\bI^{(1)}\bA^{(1)}\bx + {\boldsymbol \beta} ={\boldsymbol \Gamma} \bx + {\boldsymbol \beta}.
    \end{array}
\end{equation}

\begin{definition}[Oracle models]
\label{def:s5-oracle}
\label{def:s1-oracle}
In S5 raw-output setting, the oracle returns all logits:
$\mathcal O_{\mathrm{S5}}(\bx)=\mathcal{F}_{\theta}(\bx)$. In S1 hard-label setting, it returns
only the label
\[
    \mathcal O_{\mathrm{S1}}(\bx)=\min\operatorname*{arg\,max}_{j \in [d^{(r+2)}]}\mathcal{F}_{\theta}(\bx),
\]
where the minimum indicates a deterministic rule for ties.
\end{definition}


\begin{samepage}
\begin{definition}[Critical hyperplane, activation boundary, and critical point]
\label{def:critical-hyperplane}
\label{def:critical-point}
The critical hyperplane of neuron $(k,j)$ (the $j$-th neuron of layer $k$) in the input space of layer $k$ is
\[
    \mathbb{H}_j^{(k)}
    =\bigl\{\bu\in\mathbb R^{d^{(k)}}:
      \langle \bA_j^{(k)},\bu\rangle+\bb_j^{(k)}=0\bigr\}, k \in [r+1], j\in [d^{(k+1)}]\bigr\}.
\]
The corresponding activation
boundary in the model-input space is
\[
    \mathbb{C}_j^{(k)}
    =\bigl\{\bx^{(1)}\in\mathbb R^{d^{(1)}}: \bx^{(k)}\in\mathbb{H}_j^{(k)}\bigr\}.
\]
Therefore, $\bx^{(1)}\in\mathbb{C}_j^{(k)}$ is a critical point of the neuron $(k,j)$.
\end{definition}
\end{samepage}

\subsubsection{Prefix and suffix maps of layer $k$.}
Fix a target layer $ k \in [r+1]$. We decompose the network as
\begin{equation}\label{eq:dnn_decompose}
      \mathcal{F}_{\theta} 
    =\underbrace{f^{(r+1)}\circ\sigma\circ f^{(r)}\circ\cdots
       \circ\sigma\circ f^{(k+1)}}_{\mathcal{G}_x^{k+1}}
       \circ\sigma\circ f^{(k)}\circ
      \underbrace{\sigma\circ f^{(k-1)}\circ\cdots
       \circ\sigma\circ f^{(1)}}_{\mathcal{F}_x^{k-1}},
\end{equation}
where $\mathcal{F}_x^{k-1}$ is the layers before layer $k$, and $\mathcal{G}_x^{k+1}$ is the layers after the layer $k$. 
Given model input $\bx$ and its corresponding linear neighborhood $\mathbb{L}_{\bx}$,  
the $\mathcal{F}_x^{k-1}$ and $\mathcal{G}_x^{k+1}$ collapse to fix affine functions, {\em i.e.}, for model input $\bx^{(1)}\in \mathbb{L}_{\bx}$, and its corresponding $\bx^{(k)}$, 
\begin{equation}\label{eq:per_suffix_matrix}
    \bx^{(k)}= \mathcal{F}_{\bx}^{k-1}(\bx^{(1)})=\bF^{(k-1)}\bx^{(1)}+\bbeta^{(k-1)},
    ~
  \mathcal{G}_x^{k+1}(\bx^{(k+1)})=\bG^{(k+1)}\bx^{(k+1)}+\bgamma^{(k+1)}.  
\end{equation}
where {\small$
   \bF^{(k-1)}
      =\bI^{(k-1)}\bA^{(k-1)}\cdots \bI^{(1)}\bA^{(1)},
    ~
    \bG^{(k+1)}
      =\bA^{(r+1)}\bI^{(r)}\bA^{(r)}\cdots
        \bI^{(k+1)}\bA^{(k+1)}.
$} according to Eq. \eqref{eqn:fcnn_linear}, and hence Eq. \eqref{eqn:fcnn_linear} becomes \begin{equation}\label{eq:G_F_model_output}
\mathcal{F}_{\theta}({\bx}^{(1)}) = \bG^{(k+1)}\bI^{(k)}\left({\bA}^{(k)}\left(\bF^{(k-1)}\bx^{(1)} + \bbeta^{(k-1)}\right)+{\bb}^{(k)}\right)+\bgamma^{(k+1)}.
\end{equation}


\label{subsec:critical-decision-geometry}
\begin{definition}[Transition and dual points, decision boundary, decision facet]
\label{def:class-pair-margin}
\label{def:decision-transition}
\label{def:dual-point}
For distinct classes $a,b\in [d^{(r+2)}]$, define 
$\mathcal{D}_{ab}(\bx)=\mathcal{F}_{\theta}(\bx)_a-\mathcal{F}_{\theta}(\bx)_b$, and the set of decision boundary
\begin{equation}
\label{eq:decision boundary}
      \mathbb{D}_{ab}
    =\bigl\{\bx\in \mathbb{R}^{d^{(1)}}:\mathcal{F}_{\theta}(\bx)_a=\mathcal{F}_{\theta}(\bx)_b>\mathcal{F}_{\theta}(\bx)_c
      \text{ for every }c \in[d^{(r+2)}]/\{a,b\}\bigr\}.  
\end{equation}
Any $\bx\in \mathbb{D}_{ab}$ is a transition point for switching classes $a$ and $b$, whose decision facet in the input model space $\mathbb{R}^{d^{(1)}}$ is defined as $\mathbb{D}_{ab} \cap \mathbb{L}_{\bx} \in \mathbb{R}^{d^{(1)}}$. 
The point $\bx\in \mathbb{D}_{ab} \cap \mathbb{C}_j^{(k)}$ is a 
dual point for neuron $(k,j)$ and class pair
$(a,b)$.  Usually, there are at least two adjacent  decision facets for a given dual point $\bx$, denoted as 
 $\mathbb{D}_{ab} \cap \mathbb{L}_{\bx}$ and  $\mathbb{D}_{ab} \cap \mathbb{L}'_{\bx}$. 
\end{definition}

Figure~\ref{fig:prelim-linear-partition} summarizes different types of point in geometry, where each cell represents a  linear neighborhood. 

\begin{figure}[H]
    \centering
    \includegraphics[width=0.8\linewidth]{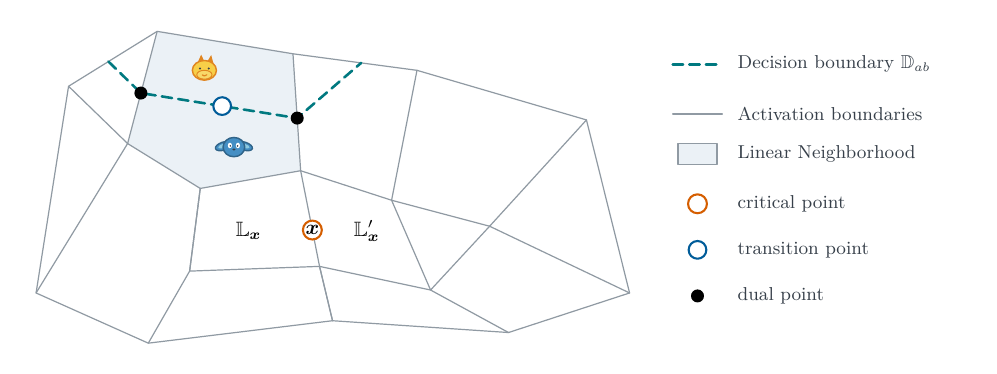}
    \caption{Piecewise-affine geometry in the model-input space.
    }
    \label{fig:prelim-linear-partition}
\end{figure}

\begin{definition}[Layer-\(k\) wiggle]
\label{def:layer-i-wiggle}
A layer-\(k\) wiggle around a model input \(\bx \) is a small vector \(\bdelta^{(k)} \in \mathbb{R}^{d^{(k)}}\), so that there exists \(\bdelta^{(1)}\in \mathbb{R}^{d^{(1)}}\) that satisfies $\bx+\bdelta^{(1)} \in \mathbb{L}_{\bx}$ (or $\bx-\bdelta^{(1)} \in \mathbb{L}_{\bx}$) and 
$\bdelta^{(k)}=\bF^{(k-1)}(\bx+\bdelta^{(1)})+\bbeta^{(k-1)}-(\bF^{(k-1)}(\bx)+\bbeta^{(k-1)})=\bF^{(k-1)}(\bdelta^{(1)})$ (or $\bdelta^{(k)}=\bF^{(k-1)}(\bx)+\bbeta^{(k-1)}- (\bF^{(k-1)}(\bx-\bdelta^{(1)})+\bbeta^{(k-1)})=\bF^{(k-1)}(\bdelta^{(1)})$). 
\end{definition}

\begin{definition}[Control space]
\label{def:control-space}
Given a model input $\bx$ and its linear neighborhood
$\mathbb L_{\bx}$, its control space at the input to layer $k$ is
\[
\mathbb V_{\bx}^{(k)}
:=
\operatorname{span}
\left\{
   \bdelta^{(k)}= \bF^{(k-1)}\bdelta^{(1)}:
    \bdelta^{(1)}\in\mathbb R^{d^{(1)}},
    \ \bx+\bdelta^{(1)}\in\mathbb L_{\bx}
\right\}
\subseteq\mathbb R^{d^{(k)}}.
\]
Equivalently, $\mathbb V_{\bx}^{(k)}$ is the subspace spanned by the
columns of $\bF^{(k-1)}$. 
\end{definition}



\begin{definition}[Projection onto the control space]
\label{def:control-space-projection}
Since the control space $\mathbb V_{\bx}^{(k)}$ is the column space of
$\bF^{(k-1)}$, its orthogonal complement satisfies $\left(\mathbb V_{\bx}^{(k)}\right)^\perp
=
\ker\left(
    \left(\bF^{(k-1)}\right)^\top
\right).$
Every vector $\bu\in\mathbb R^{d^{(k)}}$ can therefore be uniquely
written as $\bu=\overline{\bu}+\bu_\perp$, where
$\overline{\bu}\in\mathbb V_{\bx}^{(k)}$ and
$\bu_\perp\in(\mathbb V_{\bx}^{(k)})^\perp$.
We call $\overline{\bu}$ the orthogonal projection of $\bu$ onto 
$\mathbb V_{\bx}^{(k)}$ and write $\overline{\bu}:=\bP^{(k)}\bu$, where
$\bP^{(k)}$ is the corresponding orthogonal projection matrix. Equivalently, if the columns of $\bQ^{(k)}$ form an orthonormal basis
of $\mathbb V_{\bx}^{(k)}$, then
\[
\bP^{(k)}
=
\bQ^{(k)}
\bigl(\bQ^{(k)}\bigr)^\top,
\qquad
\overline{\bu}
=
\bQ^{(k)}
\bigl(\bQ^{(k)}\bigr)^\top\bu.
\]
In particular, $\bP^{(k)}$ is symmetric, i.e.,
$(\bP^{(k)})^\top=\bP^{(k)}$.
\end{definition}

\subsection{Extraction Goal and Assumptions}\label{subsec:oracle-goal}

\begin{definition}[Signature \cite{DBLP:conf/eurocrypt/CanalesMartinezCHRSS24}]
\label{def:signature}
Let
$\bA_j^{(k)}
=
({\bA}_{j,1}^{(k)},\ldots,
{\bA}_{j,d^{(k)}}^{(k)})$
be the weight vector of neuron $(k,j)$, and assume that
${\bA}_{j,1}^{(k)}\neq0$. Its signature is the vector
\begin{equation}
\label{eq:neuron-signature}
\widehat{\bA}_j^{(k)}
:=
\frac{
    \bA_j^{(k)}
}{
    {\bA}_{j,1}^{(k)}
}
=
\left(
    1,
    \frac{
        {\bA}_{j,2}^{(k)}
    }{
        {\bA}_{j,1}^{(k)}
    },
    \ldots,
    \frac{
        {\bA}_{j,d^{(k)}}^{(k)}
    }{
        {\bA}_{j,1}^{(k)}
    }
\right),
\end{equation}
where the true weight vector satisfies
$\bA_j^{(k)}
=
{\bA}_{j,1}^{(k)}
\widehat{\bA}_j^{(k)}$. 
\end{definition} 
Thus, after signature recovery, the only remaining ambiguity is the
nonzero scalar ${\bA}_{j,1}^{(k)}$. Its magnitude does not need to be
recovered \cite{DBLP:conf/crypto/CarliniJM20}: since
$\operatorname{ReLU}(cz)=c\operatorname{ReLU}(z)$ for every $c>0$,
a positive scaling can be absorbed into the outgoing weights of the
neuron. Its sign, however, is essential. Negating the recovered affine
form exchanges its active and inactive sides and cannot be absorbed
through ReLU. We call the sign of ${\bA}_{j,1}^{(k)}$ the sign of neuron
$(k,j)$. Sign recovery determines this sign and thereby identifies the
true active side of the neuron. 
Denote signs in the layer $k$ by $(s_1,s_2,\cdots,s_{d^{(k+1)}})$, $s_j \in \{1,-1\},\forall j \in [d^{(k+1)}]$, and define the sign matrix \begin{equation}\label{eqn:S_matrix}
    \bS^{(k)}=\operatorname{diag}( s^{(k)}_1, s^{(k)}_2,\dots, s^{(k)}_{d^{(k+1)}}). 
\end{equation}
Then, given $\bx\in \mathbb{R}^{d^{(1)}}$, by Eq. \eqref{eq:G_F_model_output}, the model output in S5 setting is
\begin{equation}\label{eq:sign_model_output}
\mathcal{F}_{\theta}({\bx}) = \bG^{(k+1)}\bI^{(k)}\bS^{(k)}\left(\widehat{\bA}^{(k)}\left(\bF^{(k-1)}\bx + \bbeta^{(k-1)}\right)+\widehat{\bb}^{(k)}\right)+\bgamma^{(k+1)},
\end{equation}
where $\widehat{\bA}^{(k)}$ and $\widehat{\bb}^{(k)}$ are the unsigned signatures and biases in layer $k$. 

\subsubsection{Extraction Goal.}
Our final objective is functionally equivalent parameter extraction: given
oracle access to a target network $\mathcal{F}_\theta$, recover parameters
$\hat\theta$ such that the extracted network computes the same function as
the target, up to unavoidable symmetries such as positive neuron rescaling and
permutation within a layer. This paper focuses on sign recovery in S1 hard-label setting and assumes
that a preceding signature-recovery phase has recovered the target signatures
up to nonzero scalar multiples.

\subsubsection{Assumptions.}
\begin{itemize}
    \item \textbf{Known architecture.} The attacker knows
    $(d^{(1)},\dots,d^{(r+2)})$ and that the hidden layers are fully
    connected ReLU layers.
    \item \textbf{Full-domain inputs.} The attacker may adaptively query any
    input in \(\mathbb R^{d^{(1)}}\).
    \item \textbf{Precise computation.} The analysis assumes exact real
    arithmetic or sufficiently high floating-point precision.
    \item \textbf{Oracle access.} We consider both S5 raw-output and S1
    hard-label access.
    \item \textbf{Available signatures.} For the target layer $k$, we assume that all the weights and biases of the preceding layers $1,\cdots, k-1$ are recovered, while  each neuron's
    signature $\widehat{\bA}^{(k)}$ and bias  $\widehat{\bb}^{(k)}$ are known up to an unknown nonzero scalar in layer $k$. Also, we assume that no two signatures are the same \cite{DBLP:conf/eurocrypt/CanalesMartinezCHRSS24}. Our goal is to recover the signs in layer $k$.  
\end{itemize}

%% file: sections/limitations.tex

\section{Existing Sign Recovery Methods in S1 Setting and their Limitations}
\label{sec:existing-sign-recovery}


\subsection{{\em Future Toggle} in Hard-label Setting}
\label{subsec:hard-label-future-toggle-background}

\subsubsection{{\em Neuron Wiggle}  in S5 Setting \cite{DBLP:conf/eurocrypt/CanalesMartinezCHRSS24}.}
The {\em Neuron Wiggle} method, proposed by Canales{-}Mart{\'{\i}}nez {\em et al.} at EUROCRYPT 2024, is a heuristic sign recovery method in raw-output setting. The method relies on a basic asymmetry across the target activation boundary: 
{\em the norm of the target layer output change is larger on the
target active side because ReLU blocks the target neuron's contribution
on the inactive side}. 
We first establish this asymmetry and then explain
how it motivates the observation used by {\em Future Toggle} \cite{DBLP:conf/eurocrypt/CarliniCHRS25} in S1 setting.

Suppose that $\bx$ is a critical point of the target neuron $(k,j)$, and denote its two adjacent linear neighborhoods by $\mathbb L_{\bx}$ and $\mathbb L'_{\bx}$. Their activation statuses differ only in the state of neuron $(k,j)$. Without loss of generality, suppose that the target neuron is active in $\mathbb L_{\bx}$ and inactive in $\mathbb L'_{\bx}$. Denote the corresponding activation matrices in layer $k$ by
$\bI_+^{(k)}$ and $\bI_-^{(k)}$. Choose a 
wiggle $\bdelta^{(1)}\in\mathbb R^{d^{(1)}}$ such that
$\bx+\bdelta^{(1)}\in\mathbb L_{\bx}$ and
$\bx-\bdelta^{(1)}\in\mathbb L'_{\bx}$. It induces the layer-$k$ wiggle
$\bdelta^{(k)}:=\bF^{(k-1)}\bdelta^{(1)}$ by Def. \ref{def:layer-i-wiggle}. 
Hence, the layer-$(k+1)$ wiggles are
$\bdelta^{(k+1)}
=
\bI_+^{(k)}
\bA^{(k)}
\bdelta^{(k)}$
and
$\bdelta'^{(k+1)}
=
\bI_-^{(k)}
\bA^{(k)}
\bdelta^{(k)}$
on the active and inactive sides, respectively. 
Let $\mathbb S$ denote the set of active neurons in layer $k$ within
$\mathbb L_{\bx}$. We have
\begin{equation}\label{eq:wiggle_delta_k+1}
    \|\bdelta^{(k+1)}\|^2
=
\sum_{i\in\mathbb S}
|\bA_i^{(k)}\bdelta^{(k)}|^2, ~~ ~\|\bdelta'^{(k+1)}\|^2
=
\sum_{i\in\mathbb S\setminus\{j\}}
|\bA_i^{(k)}\bdelta^{(k)}|^2.
\end{equation}
Consequently,
\begin{equation}\label{eq:wiggle_delta_k+1_gap}
\|\bdelta^{(k+1)}\|^2
=
\|\bdelta'^{(k+1)}\|^2
+
|\bA_j^{(k)}\bdelta^{(k)}|^2. 
\end{equation}
Therefore, whenever
$\bA_j^{(k)}\bdelta^{(k)}\neq0$, the layer-$(k+1)$ wiggle has
a strictly larger norm on the active side. 
For one coordinate of the model output vector ({\em e.g.}, the first coordinate), Eq. \eqref{eq:G_F_model_output} gives the following output changes:
\begin{equation}\label{eq:wiggle_changes}
\begin{array}{ll}
     \bdelta_1^{(r+2)}
&:=
(\mathcal F_{\theta}(\bx+\bdelta^{(1)})
-
\mathcal F_{\theta}(\bx))_1
=
\bG_1^{(k+1)}\bdelta^{(k+1)},\\
\bdelta_1^{\prime(r+2)}
&:=
(\mathcal F_{\theta}(\bx)
-
\mathcal F_{\theta}(\bx-\bdelta^{(1)}))_1
=
\bG_1^{(k+1)}\bdelta'^{(k+1)}.
\end{array}
\end{equation}
They satisfy
$\bdelta_1^{(r+2)}
=
\bdelta_1^{\prime(r+2)}
+
\bG_{1,j}^{(k+1)}
\bA_j^{(k)}
\bdelta^{(k)}$.
Thus, the active side contains the additional contribution of the target
neuron.

Since $|\bA_j^{(k)}\bdelta^{(k)}|
=\|\bA_j^{(k)}\|\,\|\bdelta^{(k)}\|
|\cos\angle(\bA_j^{(k)},\bdelta^{(k)})|$, $|\bA_j^{(k)}\bdelta^{(k)}|$ is maximized when $\bdelta^{(k)}$ is parallel to $\bA_j^{(k)}$. Since $\bdelta^{(k)}\in \mathbb{V}_{\bx}^{(k)}$ according to Def. \ref{def:control-space} and $\bA_j^{(k)}\not\in \mathbb{V}_{\bx}^{(k)}$, the 
{\em Neuron Wiggle} thereby chooses the direction of the wiggle $\bdelta^{(1)}\in \mathbb{R}^{d^{(1)}}$, to have a layer-$k$ wiggle $\bdelta^{(k)}$, which is exactly parallel to  the projection of $\bA_j^{(k)}$ to $\mathbb{V}_{\bx}^{(k)}$. This wiggle strengthens the target neuron $(k, j)$'s contribution, making the absolute network output change on the active side more likely to be larger. 
It therefore predicts the side with the
larger absolute output change to be the target active side
\cite{DBLP:conf/eurocrypt/CanalesMartinezCHRSS24}. 
Since a single comparison depends on several factors, notably the network architecture and the neuron activation states around $\bx$, an individual vote is not guaranteed to be correct. {\em Neuron Wiggle} thereby aggregates votes from many critical points to improve the accuracy of sign recovery. 

\subsubsection{Motivation from {\em Neuron Wiggle}: {\em Future Toggle} \cite{DBLP:conf/eurocrypt/CarliniCHRS25}.} 

In the hard-label setting, the attacker cannot observe changes in the output logits. At EUROCRYPT 2025, Carlini {\em et al.} \cite{DBLP:conf/eurocrypt/CarliniCHRS25} instead compare the walking distance from the two sides of a dual point to the first future-layer neuron toggle. When a boundary walk crosses a neuron's activation boundary, its activation state changes and the visible decision boundary bends, as illustrated in Fig.~\ref{fig:prelim-linear-partition}. We call such an activation-state change a {\em neuron toggle}. The side that reaches a future-layer toggle after a shorter distance is predicted as the target active side.

The distance comparison is motivated by the statistical signal exploited by {\em Neuron Wiggle}. Conceptually,  consider two opposite input perturbations $\bdelta^{(1)}$ and $\bdelta'^{(1)}=-\bdelta^{(1)}$ from the dual point $\bx$. The corresponding layer-$(k+1)$ wiggles $\bdelta^{(k+1)}$ and $\bdelta'^{(k+1)}$ satisfy Eq.~\eqref{eq:wiggle_delta_k+1}. For a future neuron $(p,q)$ with $p > k$, the corresponding changes in its preactivation on the two sides are $\left|
\bA_q^{(p)}\bI^{(p-1)}\cdots
\bA^{(k+1)}\bdelta^{(k+1)}
\right|$ and $\left|
\bA_q^{(p)}\bI^{(p-1)}\cdots
\bA^{(k+1)}\bdelta'^{(k+1)}
\right| $ respectively.  For the same distance $\|\bdelta^{(1)}\|$, the $\|\bdelta^{(k+1)}\|$ on the active side is larger than $\|\bdelta'^{(k+1)}\|$ on the inactive side according to Eq. \eqref{eq:wiggle_delta_k+1_gap}, hence a future neuron's preactivation tends to change faster on the active side. Although hard-label access doesn't reveal preactivation's rate of change, it reveals the resulting activation-state change when the preactivation crosses zero. {\em Future Toggle} therefore uses the distance to this toggle as an indirect proxy for the unobservable rate of change.

To strengthen this effect, the input perturbation should maximize $\left|
\bA_j^{(k)}\bF^{(k-1)}\bdelta^{(1)}
\right|.$ The ideal input-space direction is therefore parallel to $(\bA_j^{(k)}\bF^{(k-1)})^\top$. However, the attacker must remain on the visible decision boundary in order to detect its bends. For an adjacent decision facet with unit normal $\bn^{(1)}$, the walking direction is thus chosen parallel to $(\bA_j^{(k)}\bF^{(k-1)})^\top
-
\left\langle
(\bA_j^{(k)}\bF^{(k-1)})^\top,\bn^{(1)}
\right\rangle\bn^{(1)}.$ 

Starting from the dual point $\bx$, the attacker follows the projected direction on one adjacent decision facet until the decision boundary bends. Using the recovered parameters of layers $1,\ldots,k$, the attacker checks whether the bend is caused by a neuron in one of these layers. If so, it adds the current segment length(i.e., the Euclidean distance from the previous bend, or from $\bx$ for the first segment) to the accumulated distance, relocates onto the adjacent decision facet, recomputes the projected direction, and continues walking. Otherwise, the bend is attributed to a future-layer neuron toggle and the walk terminates.
The same procedure is applied on the other side of $\bx$, and the side
with the shorter accumulated distance is predicted as the target
active side.

\subsubsection{Limitations of the {\em Future Toggle}.}
\label{subsec:future-toggle-limitations}

{\em Future Toggle} \cite{DBLP:conf/eurocrypt/CarliniCHRS25} attempts to identify the target active side by comparing
the distances from a dual point to the first future layer toggles on
its two sides. However, obtaining these distances and using them for
sign recovery introduce limitations in both efficiency and accuracy.

To facilitate our discussion, we perform a sign recovery experiment using {\em Future Toggle} method~\cite{DBLP:conf/eurocrypt/CarliniCHRS25} on the CIFAR-10 network with architecture \mbox{$3072$-256$\times$3-$64$-$10$}.
The white-box information is used to identify bends caused by neurons in the recovered layers and to evaluate whether each valid vote is correct. 
Table~\ref{tab:whitebox-future-toggle-cifar10} summarizes the resulting boundary walking and voting statistics. 
The column of ``Bends from layers $1,\cdots, k$'' reports the number of bends (neuron toggles) from recovered layers encountered by walking from each attempted dual point.
Entries in Table~\ref{tab:whitebox-future-toggle-cifar10} are reported as ``the mean $\pm$ standard deviation'' values across the ten selected neurons.
\begin{table}
\centering
\caption{White-box diagnostic of {\em Future Toggle} on the CIFAR-10 network
\mbox{$3072$-$256$$\times$3-$64$-$10$}. For each hidden layer,
ten neurons are randomly selected, with $1000$ attempted dual points
evaluated for each of these neurons.}
\label{tab:whitebox-future-toggle-cifar10}
\setlength{\tabcolsep}{5.5pt}
\renewcommand{\arraystretch}{1.12}
\begin{tabular}{@{}cccc@{}}
\toprule
Layer
& \makecell{Bends from layers $1,\ldots,k$}
& \makecell{Dual point utilization (\%)\\ $\eta_{\mathrm{dual}}=n_{\mathrm{valid}}/n_{\mathrm{attempt}}$}
& \makecell{Vote accuracy (\%)\\
$p = n_{\mathrm{correct}} / n_{\mathrm{valid}}$} \\
\midrule
1 & $0.87 \pm 0.22$  & $99.74 \pm 0.18$ & $54.14 \pm 7.16$ \\
2 & $1.58 \pm 0.08$  & $94.65 \pm 1.63$ & $56.51 \pm 3.15$ \\
3 & $6.73 \pm 0.25$  & $68.18 \pm 2.63$ & $55.31 \pm 2.31$ \\
4 & $12.45 \pm 0.52$ & $3.21 \pm 0.53$  & $52.82 \pm 8.08$ \\
\bottomrule
\end{tabular}

\end{table}

\begin{itemize}
    \item {\bf Limitation 1: Cost of decision boundary tracing.}
During each of the two walks  in $\mathbb{D}_{ab} \cap \mathbb{L}_{\bx}$ and $\mathbb{D}_{ab} \cap \mathbb{L}'_{\bx}$ from a dual point $\bx$, the attacker must detect when the current decision facet ends and the decision boundary bends. This information is not directly provided by the hard-label oracle. Instead, the attacker must repeatedly query the oracle to determine whether the walk remains on the same decision facet and use binary search to locate the bend when the facet changes. 
This cost is further amplified when a detected bend is caused by the neuron in a known layer $1,\ldots,k$ rather than by a future layer. Such a bend does not terminate the walk. Instead, the attacker must recover the normal of the new decision facet and continue the walk. 
Each such bend requires the recovery of a new decision facet normal, which involves $d^{(1)}-1$ coordinate ratio searches, according to Eq.~\eqref{eq:recovered-input-space-normal} in Sect. \ref{subsec:recovering-projected-unit-normals}. Consequently, this cost can be substantial when the model input dimension is large ({\em e.g.,} $d^{(1)}=3072$).

In the 2nd column of Table~\ref{tab:whitebox-future-toggle-cifar10}, when the layer depth increases, the mean number of bends from the recovered layers rises from $0.87$ to $12.45$, indicating that walks in deeper layers require more decision-facet normal recoveries.

\item {\bf Limitation 2: Limited utilization of dual points.} 
A search may repeatedly encounter activation boundaries belonging to recovered layers or fail to reach a future layer toggle within the maximum searching distance. In either case, the dual point is discarded.
Let $n_{\mathrm{attempt}}$ denote the number of
attempted dual points and $n_{\mathrm{valid}}$ the number that produce
valid votes. The dual point utilization rate is defined as
$\eta_{\mathrm{dual}}:=n_{\mathrm{valid}}/n_{\mathrm{attempt}}$. 
A lower $\eta_{\mathrm{dual}}$ requires more attempted dual points and results in a longer time.

In the third column of Table~\ref{tab:whitebox-future-toggle-cifar10}, as the layer depth increases, the dual point utilization rate falls from $99.74\%$ to $3.21\%$. 

\item {\bf Limitation 3: Weak statistical advantage.}
{\em Future Toggle} infers the target active side through two successive
proxy relations: a larger target layer output change is expected to
produce faster changes in future neurons, and the faster neuron value changes are expected to produce a shorter distance to the first future layer toggle. 
Neither relation is guaranteed to hold at every dual point, and 
the probability that the target active side produces the shorter distance may be only slightly greater than one half.

Let $p=1/2+\gamma$ denote the probability that an individual valid vote is correct, where $\gamma>0$ represents the statistical advantage in favor of the target active side. Assuming that the valid votes are independent and share the same success probability $p$, Hoeffding's inequality gives
$\Pr[\textnormal{the majority vote is incorrect}]
\leq\exp(-2\gamma^2n_{\mathrm{valid}})$.
Therefore, ensuring an error probability of at most $\alpha$ requires
$n_{\mathrm{valid}}
\geq\ln(1/\alpha)/(2\gamma^2)$, which grows with $1/\gamma^2$. For example, when $p=0.55$, about $n_{\mathrm{valid}}=10^3$ valid votes are required to reach a confidence level of $99\%$ \cite{DBLP:conf/eurocrypt/CarliniCHRS25}. 
This bound shows that a weaker statistical advantage requires more
valid votes to reach the same confidence level. With a limited budget
$n_\mathrm{attempt}$, more neurons may therefore remain below the required confidence
threshold. 

In the 4th column of Table~\ref{tab:whitebox-future-toggle-cifar10}, $n_{\mathrm{correct}}$ denotes the number of
valid dual points that produce correct votes; the vote accuracy is $p = n_{\mathrm{correct}} / n_{\mathrm{valid}}$. The average vote accuracy of each layer ranges from $52.82\%$ to $56.51\%$.  It shows that an individual valid vote is only slightly more likely to be correct than random guessing.

\item {\bf Limitation 4: Amplified limitations in the last hidden layer.}
The {\em Future Toggle} requires a different terminal event for the last hidden layer,  
where there is
no future-layer ReLU neuron.  
It therefore continues each search until an intersection of decision boundaries  and uses this distance instead \cite{DBLP:conf/eurocrypt/CarliniCHRS25}.

This special treatment amplifies the preceding three limitations.
First, a class decision boundary intersection may be farther from the
dual point than a future layer neuron toggle, resulting in a longer walk. Second, the required intersection may not be reached within the maximum walking distance and the corresponding dual point is
then discarded, reducing $\eta_{\mathrm{dual}}$ and increasing the
number $n_{\mathrm{attempt}}$ of attempted dual points. Finally, the two
walks may terminate at intersections with different class decision boundaries, so their measured distances depend on different class decision boundaries as well as on the rates of logit change. 

In the last hidden layer of Table~\ref{tab:whitebox-future-toggle-cifar10}, only about $32$ of the $1000$
attempted dual points per neuron produce valid votes on average.
\end{itemize}

\subsection{System of Equations ({\em SOE}) and Its Extensions}
\label{subsec:SOE and its extensions}


\subsubsection{The Raw-Output {\em SOE} Method \cite{DBLP:conf/eurocrypt/CanalesMartinezCHRSS24}.} At EUROCRYPT 2024, Canales{-}Mart{\'{\i}}nez {\em et al.} proposed the System of Equations ({\em SOE}) method in S5 setting. 
Specifically, targeting a single coordinate of the output vector ({\em e.g.}, the first coordinate) at an input $\bx$, the attacker samples $d^{(k+1)}$ perturbations $\bdelta^{(1)}_1,\ldots,\bdelta^{(1)}_{d^{(k+1)}} \in \mathbb{R}^{d^{(1)}}$ such that every perturbed input $\bx+\bdelta^{(1)}_\ell$ ($\ell \in [d^{(k+1)}]$) remains within the same linear neighborhood as $\bx$. Then the linear system is built by computing the output difference between $\mathcal F_{\theta}(\bx+\bdelta_{\ell})$ and $\mathcal F_{\theta}(\bx)$, 
\begin{equation}\small
\label{eq:raw-output-SOE}
\begin{bmatrix}
\left(
    \widehat{\bA}^{(k)}
    \bF^{(k-1)}
    \bdelta^{(1)}_1
\right)^\top
\\
\vdots
\\
\left(
    \widehat{\bA}^{(k)}
    \bF^{(k-1)}
    \bdelta^{(1)}_{d^{(k+1)}}
\right)^\top
\end{bmatrix}
\left(
    \bG_1^{(k+1)}
    \bI^{(k)}\bS^{(k)}
\right)^\top
=
\begin{bmatrix}
(
    \mathcal F_{\theta}(\bx+\bdelta^{(1)}_1)
    -
    \mathcal F_{\theta}(\bx)
)_1
\\
\vdots
\\
(
    \mathcal F_{\theta}
    (\bx+\bdelta^{(1)}_{d^{(k+1)}})
    -
    \mathcal F_{\theta}(\bx)
)_1
\end{bmatrix},
\end{equation}
where $\bF^{(k-1)}$, $\widehat{\bA}^{(k)}$, $\widehat{\bb}^{(k)}$ are known and the target is to recover $\bS^{(k)}$. 
Solving this linear system yields the  unknown vector $(\bG_1^{(k+1)} \bI^{(k)}\bS^{(k)})^\top$.
Since the ReLU activation suppresses negative values, any neuron $t$ in the layer $k$ that is inactive at $\bx$ will have a corresponding entry of zero in $(\bG_1^{(k+1)} \bI^{(k)}\bS^{(k)})^\top$.
Once the inactive neurons are identified ({\em e.g.}, neuron $t$), if $\widehat{\bA}^{(k)}_t \bx^{(k)} + \widehat{\bb}^{(k)}_t>0$, its sign will be $s^{(k)}_t=-1$. 
To ensure the linear system in Eq.~(\ref{eq:raw-output-SOE}) has a unique solution, the coefficient matrix must be of full rank, which requires $\text{rank}(\widehat{\bA}^{(k)}\bF^{(k-1)}) = d^{(k+1)}$. This implies $d^{(1)}, \dots, d^{(k)} \geq  d^{(k+1)}$. 
Consequently, this deterministic method is primarily applicable to network architectures that are sufficiently contractive.

\subsubsection{The {\em SOE} + {\em Wiggle} in Raw-Output Setting \cite{DBLP:conf/eurocrypt/LiuSELBP26}.} 
\label{subsubsec:removing-inactive-neurons}
At EUROCRYPT 2026, Liu {\em et al.}  combined the {\em SOE} and the {\em Neuron Wiggle} methods. As summarized in Sect.~\ref{subsec:hard-label-future-toggle-background}, {\em Neuron Wiggle} \cite{DBLP:conf/eurocrypt/CanalesMartinezCHRSS24} recovers the sign of each neuron along with a confidence level, where a high confidence level strongly indicates a correct recovery. Recall that the unknown vector in Eq. \eqref{eq:raw-output-SOE} is expressed as ${(\bG^{(k+1)}_1 \bI^{(k)} \bS^{(k)})}^\top = (\bG^{(k+1)}_{1,1}\tau^{(k)}_1s^{(k)}_1, \dots, \bG^{(k+1)}_{1,d^{(k+1)}}\tau^{(k)}_{d^{(k+1)}}s^{(k)}_{d^{(k+1)}})^\top$. If neuron $t$ is identified as inactive (determined by evaluating $s^{(k)}_t (\widehat{\bA}^{(k)}_t \bx^{(k)} + \widehat{\bb}^{(k)}_t)<0$, provided that the sign $s^{(k)}_t$ is recovered correctly with a high confidence level by {\em Neuron Wiggle}), the corresponding term $\bG^{(k+1)}_{1,t}\tau^{(k)}_t s^{(k)}_t$ can be directly set to $0$, thereby reducing the number of unknowns.
Let $\mathbb K$ contain the neurons that {\em Neuron Wiggle} identifies as
inactive at $\bx$ with high confidence, and
$\mathbb U:=[d^{(k+1)}]\setminus\mathbb K$ contain the remaining neurons.
For a vector $\boldsymbol v$, let
$[\boldsymbol v]_{\mathbb K}$ and
$[\boldsymbol v]_{\mathbb U}$ denote the subvectors indexed by
$\mathbb K$ and $\mathbb U$, respectively. Since every neuron in
$\mathbb K$ is inactive, the corresponding entries of
$(\bG_1^{(k+1)}\bI^{(k)} \bS^{(k)})^\top$ are zero. The {\em SOE} thereby reduces to
\begin{equation}\small
\label{eq:reduced-raw-output-SOE}
\begin{bmatrix}
\left[
    \widehat{\bA}^{(k)}
    \bF^{(k-1)}
    \bdelta^{(1)}_1
\right]_{\mathbb U}^{\top}
\\
\vdots
\\
\left[
    \widehat{\bA}^{(k)}
    \bF^{(k-1)}
     \bdelta^{(1)}_{d^{(k+1)}}
\right]_{\mathbb U}^{\top}
\end{bmatrix}
\left[
    \left(
        \bG_1^{(k+1)}
        \bI^{(k)}
        \bS^{(k)}
    \right)^\top
\right]_{\mathbb U}
=
\begin{bmatrix}
(
    \mathcal F_{\theta}(\bx+\bdelta^{(1)}_1)
    -
    \mathcal F_{\theta}(\bx)
)_1
\\
\vdots
\\
(
    \mathcal F_{\theta}(\bx+\bdelta^{(1)}_{d^{(k+1)}})
    -
    \mathcal F_{\theta}(\bx)
)_1
\end{bmatrix}.
\end{equation}
The total number of unknowns reduces from $d^{(k+1)}$ to $|\mathbb U|$. For these remaining unknowns to be uniquely determined, the rank of the coefficient matrix in Eq.~\eqref{eq:reduced-raw-output-SOE} must equal  $|\mathbb U|$. Since the rank of this coefficient matrix is upper-bounded by $\operatorname{rank}(\widehat{\bA}^{(k)}\bF^{(k-1)})$, then $|\mathbb U| \leq \operatorname{rank}(\widehat{\bA}^{(k)}\bF^{(k-1)})$. Since $|\mathbb U| + |\mathbb K| = d^{(k+1)}$, this yields the following necessary condition for uniquely solving Eq.~\eqref{eq:reduced-raw-output-SOE}:
\begin{equation}
\label{eq:reduced-SOE-necessary-condition}
|\mathbb K|
\geq
d^{(k+1)}
-
\operatorname{rank}
\left(
    \widehat{\bA}^{(k)}
    \bF^{(k-1)}
\right).
\end{equation}

The original {\em SOE} method \cite{DBLP:conf/eurocrypt/CanalesMartinezCHRSS24} constructs the linear system at a single input $\bx$, which inherently limits the rank of the system to the rank of the local prefix matrix $\bF^{(k-1)}$ around $\bx$. 
However, since the unknown vector is $(\bG_1^{(k+1)} \bI^{(k)} \bS^{(k)})^\top$, any input point sharing the same activation state in the layers $k$ to $r+1$ 
can be utilized.  
In other words, as long as the activation states of all neurons in layer $k$ and all subsequent layers remain invariant, the target layer's activation matrix $\bI^{(k)}$ and the suffix map $\bG^{(k+1)}$ are identical across these inputs. 
Meanwhile, the activation states of neurons in layers
$1,\ldots,k-1$ may differ across these inputs. These differences can
produce distinct prefix maps
$\bF_{\bx_t}^{(k-1)}$ for $\bx_t$ ($1 \leq t \leq T$),
and hence distinct {\em SOE} coefficient matrices. 
At each $\bx_t$ ($1 \leq t \leq T$), choose perturbations
$\bdelta^{(1)}_{t,1},\ldots,\bdelta^{(1)}_{t,d^{(k+1)}}$
that remain in its linear neighborhood. After removing the entries
indexed by $\mathbb K$, the corresponding {\em SOE} systems can be combined
as
\begin{equation}\small
\label{eq:raw-output-SOE-extension}
\begin{bmatrix}
\left[
    \widehat{\bA}^{(k)}
    \bF_{\bx_1}^{(k-1)}
    \bdelta^{(1)}_{1,1}
\right]_{\mathbb U}^{\top}
\\
\vdots
\\
\left[
    \widehat{\bA}^{(k)}
    \bF_{\bx_1}^{(k-1)}
   \bdelta^{(1)}_{1,d^{(k+1)}}
\right]_{\mathbb U}^{\top}
\\
\vdots
\\
\left[
    \widehat{\bA}^{(k)}
    \bF_{\bx_T}^{(k-1)}
    \bdelta^{(1)}_{T,1}
\right]_{\mathbb U}^{\top}
\\
\vdots
\\
\left[
    \widehat{\bA}^{(k)}
    \bF_{\bx_T}^{(k-1)}
    \bdelta^{(1)}_{T,d^{(k+1)}}
\right]_{\mathbb U}^{\top}
\end{bmatrix}
\left[
    \left(
        \bG_1^{(k+1)}
        \bI^{(k)}
        \bS^{(k)}
    \right)^\top
\right]_{\mathbb U}
=
\begin{bmatrix}
(
    \mathcal F_{\theta}(\bx_1+\bdelta^{(1)}_{1,1})
    -
    \mathcal F_{\theta}(\bx_1)
)_1
\\
\vdots
\\
(
    \mathcal F_{\theta}
    (\bx_1+\bdelta^{(1)}_{1,d^{(k+1)}})
    -
    \mathcal F_{\theta}(\bx_1)
)_1
\\
\vdots
\\
(
    \mathcal F_{\theta}(\bx_T+\bdelta^{(1)}_{T,1})
    -
    \mathcal F_{\theta}(\bx_T)
)_1
\\
\vdots
\\
(
    \mathcal F_{\theta}
    (\bx_T+\bdelta^{(1)}_{T,d^{(k+1)}})
    -
    \mathcal F_{\theta}(\bx_T)
)_1
\end{bmatrix}.
\end{equation}

The combined coefficient matrix may have a higher rank than the
coefficient matrix obtained at any individual input. When its rank
reaches $|\mathbb U|$, the remaining unknown entries are uniquely
determined.
\subsubsection{Hard-Label {\em SOE} \cite{DBLP:conf/latincrypt/CanalesMartinezS25}.} Canales{-}Mart{\'{\i}}nez {\em et al.} extended the {\em SOE} method \cite{DBLP:conf/eurocrypt/CanalesMartinezCHRSS24} to S1 setting at LATINCRYPT 2025. 
Suppose a transition point
$\bx\in\mathbb D_{ab}$ is located at the decision boundary between classes $a$ and $b$ ($a < b$). According to Eq.~\eqref{eq:decision boundary} and \eqref{eq:per_suffix_matrix}, we have:
\begin{equation}\label{eq:hard_lable_SOE_decision}
\mathcal D_{ab}(\bx)
=
\left(
    \bG_a^{(k+1)}-\bG_b^{(k+1)}
\right)  \bx^{(k+1)}
+
\bgamma_a^{(k+1)}-\bgamma_b^{(k+1)} = 0,
\end{equation}
where $ \bx^{(k+1)} =\bI^{(k)} (\bA^{(k)} (\bF^{(k-1)}\bx+\bbeta^{(k-1)})+\bb^{(k)})$. 
The attacker then samples $d^{(k+1)}-1$ perturbations $\bdelta^{(1)}_1,\ldots,\bdelta^{(1)}_{d^{(k+1)}-1}$ such that $\bx+\bdelta^{(1)}_\ell \in \mathbb{D}_{ab}\cap \mathbb{L}_{\bx}$, with $\ell \in [d^{(k+1)}-1]$.  
Then according to Eq. \eqref{eq:hard_lable_SOE_decision}, we have 
\begin{equation}
    \begin{aligned}
\mathcal D_{ab}(\bx+\bdelta^{(1)}_\ell) - \mathcal{D}_{ab}(\bx)
= 
\left(
    \bG_a^{(k+1)}-\bG_b^{(k+1)}
\right)
\bI^{(k)} \bS^{(k)} \widehat{\bA}^{(k)} \bF^{(k-1)} \bdelta^{(1)}_\ell
= 0.
\end{aligned}
\end{equation}
With $d^{(k+1)}-1$ linearly independent equations from $\bdelta^{(1)}_1,\ldots,\bdelta^{(1)}_{d^{(k+1)}-1}$, we can construct
the linear system,
\begin{equation}
\label{eq:hard-label SOE}
\begin{bmatrix}
\left[
    \widehat{\bA}^{(k)}
    \bF^{(k-1)}
    \bdelta_1^{(1)}
\right]^{\top}
\\
\vdots
\\
\left[
    \widehat{\bA}^{(k)}
    \bF^{(k-1)}
    \bdelta^{(1)}_{d^{(k+1)}-1}
\right]^{\top}
\end{bmatrix}
\left[
    \left(
        (\bG_a^{(k+1)}-\bG_b^{(k+1)})
        \bI^{(k)}
        \bS^{(k)}
    \right)^\top
\right]
=
\bm0.
\end{equation}
Solving Eq.~\eqref{eq:hard-label SOE} yields the unknown vector  
$\left(
    (\bG_a^{(k+1)}-\bG_b^{(k+1)})
    \bI^{(k)}
    \bS^{(k)}
\right)^\top$. 
Then the sign recovery process is similar to {\em SOE} method introduced above. 
\paragraph{Limitation.} To ensure the linear system in Eq.~\eqref{eq:hard-label SOE} has a unique nonzero solution up to a scalar multiple, also the same as the {\em SOE} method, hard label {\em SOE} requires
$\operatorname{rank}(\widehat{\bA}^{(k)}\bF^{(k-1)})=d^{(k+1)}$, then 
$d^{(1)}, \dots, d^{(k)} \geq d^{(k+1)}$. 
Consequently, the network architectures also need to be {\em strongly contractive}.

%% file: sections/hardlabel.tex
\section{Hard-Label Sign Recovery via \em{Normal Alignment}}
\label{sec:hard-label-normal-jump}





\subsection{Recovering Decision-Facet Normals in S5/S1 Settings}
\label{subsec:recovering-projected-unit-normals}

\subsubsection{Recovering Unit Decision-Facet Normals in the Model Input Space Using  \cite{DBLP:conf/asiacrypt/ChenDGSWW24,DBLP:conf/eurocrypt/CarliniCHRS25}.}
\label{subsubsec:recovering-input-space-unit-normal}
Given a transition point $\bx\in\mathbb D_{ab}$, recall from Eq.~\eqref{eq:hard_lable_SOE_decision}, we have $\mathcal D_{ab}(\bx)
=
\left(
    \bG_a^{(1)}-\bG_b^{(1)}
\right)  \bx^{(1)}
+
\bgamma_a^{(1)}-\bgamma_b^{(1)} = 0$.
The normal vector $\boldsymbol g^{(1)}$ of the local decision facet $\mathbb{D}_{ab}\cap \mathbb{L}_{\bx}$, is given by
$\boldsymbol g^{(1)} := (\bG_a^{(1)}-\bG_b^{(1)})^\top\in \mathbb{R}^{d^{(1)}}$.

The $\boldsymbol g^{(1)}$ can be recovered up to a nonzero
scalar using hard-label queries following the methods in \cite{DBLP:conf/asiacrypt/ChenDGSWW24,DBLP:conf/eurocrypt/CarliniCHRS25}. Let $\boldsymbol e_1,\ldots,\boldsymbol e_{d^{(1)}}$ denote the
standard basis of the model input space $\mathbb R^{d^{(1)}}$.
For each $\ell\in[d^{(1)}]$, let
$\boldsymbol g^{(1)}_\ell
:=
\langle\boldsymbol g^{(1)},\boldsymbol e_\ell\rangle$
denote the $\ell$-th coordinate of $\boldsymbol g^{(1)}$.
Suppose that $\boldsymbol g^{(1)}_1\neq0$.
For each $\ell\in[d^{(1)}]$, take a sufficiently small
step $\alpha\boldsymbol e_\ell$ and use hard-label binary search along
$\boldsymbol e_1$ to find a scalar $\beta_\ell$ such that
$\bx+\alpha\boldsymbol e_\ell+\beta_\ell\boldsymbol e_1$ returns to the decision boundary $\mathbb D_{ab}$. Provided that the returned point remains on the same affine decision
facet, then
we have $\alpha\boldsymbol g^{(1)}_\ell+\beta_\ell\boldsymbol g^{(1)}_1=0$, and hence
$\boldsymbol g^{(1)}_\ell/\boldsymbol g^{(1)}_1=-\beta_\ell/\alpha$.
Using these recovered coordinate ratios, we can recover:
\begin{equation}
\label{eq:recovered-input-space-normal}
\widehat{\boldsymbol g}^{(1)}
:=
\left(
    1,
    -\frac{\beta_2}{\alpha},
    \ldots,
    -\frac{\beta_{d^{(1)}}}{\alpha}
\right)
=
\frac{\boldsymbol g^{(1)}}{\boldsymbol g^{(1)}_1}.
\end{equation} 
Although its magnitude cannot be recovered, its sign can be determined with the following method. 
For a sufficiently small $\varepsilon>0$, $\bx+\varepsilon\widehat{\boldsymbol g}^{(1)}$ is still in the linear neighborhood $\mathbb L_{\bx}$, 
we have
\begin{equation}
\mathcal D_{ab}(\bx+\varepsilon\widehat{\boldsymbol g}^{(1)})
-\mathcal D_{ab}(\bx)
= \boldsymbol g^{(1)} \cdot \varepsilon\widehat{\boldsymbol g}^{(1)}
=\varepsilon \langle\boldsymbol g^{(1)},\widehat{\boldsymbol g}^{(1)}\rangle.
\end{equation}
\begin{itemize}
    \item If the hard-label oracle returns class $a$ at $\bx+\varepsilon\widehat{\boldsymbol g}^{(1)}$, then $\mathcal D_{ab}(\bx+\varepsilon\widehat{\boldsymbol g}^{(1)}) -\mathcal D_{ab}(\bx)>0$, {\em i.e.,} $\langle\boldsymbol g^{(1)},\widehat{\boldsymbol g}^{(1)}\rangle>0$. The sign of $\widehat{\boldsymbol g}^{(1)}$ follows $\boldsymbol g^{(1)}$.
    \item In contrast, if it returns class $b$, then $\langle\boldsymbol g^{(1)},\widehat{\boldsymbol g}^{(1)}\rangle<0$, and $\widehat{\boldsymbol g}^{(1)}$ needs to be reversed.
\end{itemize}
After normalization, we obtain the {\em unit decision-facet normal} of $\mathbb{D}_{ab}\cap \mathbb{L}_{\bx}$ as
$\boldsymbol n^{(1)}:=\boldsymbol g^{(1)}/\|\boldsymbol g^{(1)}\|$.

\subsubsection{Projected Decision-Facet Normal in the Input Space of Layer $k$.}
\label{subsubsec:recovering-layer-space-unit-normal}

Let $\boldsymbol g^{(k)}$ denote the {\em decision-facet normal in the input space of layer $k\geq 1$}, then $\boldsymbol g^{(k)}=
\left[
    \left(
        \bG_a^{(k+1)}-\bG_b^{(k+1)}
    \right)
    \bI^{(k)}
    \bA^{(k)}
\right]^\top
\in\mathbb R^{d^{(k)}}.$
Note that the decision-facet normal $\boldsymbol g^{(1)} $ in the model input space and the decision-facet normal $\boldsymbol g^{(k)}$ in the input space of
layer $k$ satisfy:
\begin{equation}
    \label{eq:g and gk}
    \boldsymbol g^{(1)}=(\bF^{(k-1)})^\top \boldsymbol g^{(k)}.
\end{equation}
This constructs a linear system with an unknown vector $\boldsymbol g^{(k)}$. Then we can recover $\boldsymbol g^{(k)}$ from $ \boldsymbol g^{(1)}$ by solving this system. If
$\operatorname{rank}(\bF^{(k-1)})=d^{(k)}$,
then the system has the unique solution $\boldsymbol g^{(k)}$.
When
$\operatorname{rank}(\bF^{(k-1)})<d^{(k)}$,
the system admits multiple solutions, and its solution set is
$\boldsymbol g^{(k)}
+\ker((\bF^{(k-1)})^\top)$.
Although the full $\boldsymbol g^{(k)}$ is not uniquely determined in the latter
case, its projection
$\overline{\boldsymbol g}^{(k)}
=\bP^{(k)}\boldsymbol g^{(k)}$
can be recovered as the unique minimum-norm solution of the system
using least squares as proved below. 

Let $\boldsymbol h$ be any solution to
Eq.~\eqref{eq:g and gk}, so that
$\boldsymbol h=\boldsymbol g^{(k)}+\boldsymbol u$ for some
$\boldsymbol u\in\ker((\bF^{(k-1)})^\top)$.
By Def.~\ref{def:control-space-projection},
$\ker((\bF^{(k-1)})^\top)
=(\mathbb V_{\bx}^{(k)})^\perp$.
Since $\bP^{(k)}$ is the orthogonal projector onto
$\mathbb V_{\bx}^{(k)}$, we have
$\boldsymbol g^{(k)}-\bP^{(k)}\boldsymbol g^{(k)}
\in\ker((\bF^{(k-1)})^\top)$.
Therefore, we can rewrite $\boldsymbol h$ as
\begin{equation}
\label{eq:solution-projection-decomposition}
\begin{aligned}
\boldsymbol h
&=
\bP^{(k)}\boldsymbol g^{(k)}
+
\left(
\boldsymbol g^{(k)}
-\bP^{(k)}\boldsymbol g^{(k)}
+\boldsymbol u
\right)
=
\bP^{(k)}\boldsymbol g^{(k)}
+\boldsymbol u',
\end{aligned}
\end{equation}
where
$\boldsymbol u'
=\boldsymbol g^{(k)}
-\bP^{(k)}\boldsymbol g^{(k)}
+\boldsymbol u
\in\ker((\bF^{(k-1)})^\top)$. Since $\bP^{(k)}\boldsymbol g^{(k)}
\in\mathbb V_{\bx}^{(k)}$ and
$\boldsymbol u'\in(\mathbb V_{\bx}^{(k)})^\perp$,
the two vectors are orthogonal. Consequently,
$\lVert\boldsymbol h\rVert^2
=\lVert\bP^{(k)}\boldsymbol g^{(k)}\rVert^2
+\lVert\boldsymbol u'\rVert^2$.
The norm is therefore uniquely minimized when
$\boldsymbol u'=\boldsymbol 0$.
Thus, the projected decision-facet normal
$\overline{\boldsymbol g}^{(k)}
=\bP^{(k)}\boldsymbol g^{(k)}$
is the unique minimum-norm solution of
Eq.~\eqref{eq:g and gk}.

In the hard-label setting, the attacker recovers the {\em unit decision-facet normal}
$\boldsymbol n^{(1)}=\boldsymbol g^{(1)}/\|\boldsymbol g^{(1)}\|$ rather than $\boldsymbol g^{(1)}$. We therefore solve
$(\bF^{(k-1)})^\top \frac{ \boldsymbol g^{(k)}}{\|\boldsymbol g^{(1)}\|}
=\boldsymbol n^{(1)}$
and choose any solution
$\frac{ \boldsymbol g^{(k)}}{\|\boldsymbol g^{(1)}\|}\in\mathbb R^{d^{(k)}}$.
Following the same argument as above, all such solutions have the same
projection onto the control space. Given any solution 
$\boldsymbol g^{(k)}/\|\boldsymbol g^{(1)}\|$, this common
projection is
$\bP^{(k)}\frac{ \boldsymbol g^{(k)}}{\|\boldsymbol g^{(1)}\|}
=\frac{\overline{\boldsymbol g}^{(k)}}{\|\boldsymbol g^{(1)}\|}$.
Normalizing this projection gives 
\begin{equation}\label{eq:n_bar}
\overline{\boldsymbol n}^{(k)}
:=
\frac{
    \overline{\boldsymbol g}^{(k)}
}{
    \|\overline{\boldsymbol g}^{(k)}\|
}.
\end{equation}

\subsection{Normal Lengths Comparison in Raw-Output Setting}
\label{subsec:normal-length-model}

\subsubsection{Two-Side Projected Decision-Facet Normals.} 
Let $\bx\in\mathbb C_j^{(k)}\cap\mathbb D_{ab}$ be a dual point, and
let $\mathbb L_{\bx}$ and $\mathbb L'_{\bx}$ denote the two linear
neighborhoods adjacent to the target critical hyperplane.
Then the neuron activation states in $\mathbb L_{\bx}$ and $\mathbb L'_{\bx}$ differ only in the state of neuron $(k,j)$. Without loss of generality, suppose that the target neuron is active in $\mathbb L_{\bx}$ and inactive in $\mathbb L'_{\bx}$. 
Denote the corresponding activation matrices by $\bI_+^{(k)}$ and $\bI_-^{(k)}$, and the decision-facet normals in the input space of layer $k$ by $\boldsymbol g^{(k)}$ and $\boldsymbol g^{\prime(k)}$, respectively. Then we have,
\begin{equation}
\resizebox{\linewidth}{!}{$
\label{eq:two-sided-layer-margin-normals}
\begin{aligned}
\boldsymbol g^{(k)}
&:=
\left[
    \left(
        \bG_a^{(k+1)}-\bG_b^{(k+1)}
    \right)
    \bI_+^{(k)}
    \bA^{(k)}
\right]^\top\\
&=\sum_{i \in \mathbb{S}/\{j\}}(\bG_{a,i}^{(k+1)}-\bG_{b,i}^{(k+1)})
(\bA_i^{(k)})^\top + (\bG_{a,j}^{(k+1)}-\bG_{b,j}^{(k+1)})
(\bA_j^{(k)})^\top,\\
\boldsymbol g^{\prime(k)}
&:=
\left[
    \left(
        \bG_a^{(k+1)}-\bG_b^{(k+1)}
    \right)
    \bI_-^{(k)}
    \bA^{(k)}
\right]^\top =\sum_{i \in \mathbb{S}/\{j\}}(\bG_{a,i}^{(k+1)}-\bG_{b,i}^{(k+1)})
(\bA_i^{(k)})^\top,
\end{aligned}
$}
\end{equation}
where $\mathbb{S}$ contains the indices of all active neurons in layer $k$. 
Therefore, the difference between $\bg^{(k)}$ and $\bg^{\prime(k)}$ is exactly a scalar multiple of the target neuron's weight $A^{(k)}_j$, {\em i.e.}, 
\begin{equation}\label{eq:gap}
\boldsymbol g^{(k)}-\boldsymbol g^{\prime(k)}
=(\bG_{a,j}^{(k+1)}-\bG_{b,j}^{(k+1)})
(\bA_j^{(k)})^\top.
\end{equation}

According to Eq.~(\ref{eq:g and gk}) in Sect.~\ref{subsec:recovering-projected-unit-normals}, $\boldsymbol g^{(k)}$ (or $\boldsymbol g^{\prime(k)}$) can be recovered from the system $\boldsymbol g^{(1)} =(\bF^{(k-1)})^\top \boldsymbol g^{(k)}$. 
In the deeper hidden layers, $\bF^{(k-1)}$ is often rank deficient, so the
full $\boldsymbol g^{(k)}$ (or $\boldsymbol g^{\prime(k)}$)  is generally not uniquely recoverable. Only their
projections onto the control space $\mathbb V_{\bx}^{(k)}$, 
$\overline{\boldsymbol g}^{(k)}
:=\bP^{(k)}\boldsymbol g^{(k)}$ and
$\overline{\boldsymbol g}^{\prime(k)}
:=\bP^{(k)}\boldsymbol g^{\prime(k)}$,
are uniquely determined.

By Def.~\ref{def:control-space-projection}, $\bP^{(k)}$ is
symmetric. Hence, the projected decision-facet normals can be written as
\begin{equation}\small
\label{eq:two-sided-projected-layer-margin-normals}
\begin{aligned}
\overline{\boldsymbol g}^{(k)}
&:=
\bP^{(k)}\boldsymbol g^{(k)}
=
\left[
    \left(
        \bG_a^{(k+1)}-\bG_b^{(k+1)}
    \right)
    \bI_+^{(k)}
    \bA^{(k)}
    \bP^{(k)}
\right]^\top
\\
&=
\sum_{i\in \mathbb{S}/\{j\}}
\left(
    \bG_{a,i}^{(k+1)}-\bG_{b,i}^{(k+1)}
\right)
\bP^{(k)}
\left(\bA_i^{(k)}\right)^\top
+
\left(
    \bG_{a,j}^{(k+1)}-\bG_{b,j}^{(k+1)}
\right)
\bP^{(k)}
\left(\bA_j^{(k)}\right)^\top,
\\
\overline{\boldsymbol g}^{\prime(k)}
&:=
\bP^{(k)}\boldsymbol g^{\prime(k)}
=
\left[
    \left(
        \bG_a^{(k+1)}-\bG_b^{(k+1)}
    \right)
    \bI_-^{(k)}
    \bA^{(k)}
    \bP^{(k)}
\right]^\top\\
&=
\sum_{i\in \mathbb{S}/\{j\}}
\left(
    \bG_{a,i}^{(k+1)}-\bG_{b,i}^{(k+1)}
\right)
\bP^{(k)}
\left(\bA_i^{(k)}\right)^\top.
\end{aligned}
\end{equation}

\subsubsection{Statistical Length Advantage.}
\label{subsubsec:length-based-motivation}
For a matrix $\boldsymbol M$, its squared Frobenius norm is defined by
$\|\boldsymbol M\|_F^2
:=
\sum_m\|\boldsymbol M_{m}\|^2
=
\sum_m\sum_{n}\boldsymbol M_{m,n}^2$.
The matrices
$\boldsymbol I_+^{(k)}
\boldsymbol A^{(k)}
\boldsymbol P^{(k)}$
and
$\boldsymbol I_-^{(k)}
\boldsymbol A^{(k)}
\boldsymbol P^{(k)}$
have identical rows except for the row corresponding to the target
neuron, which is
$\boldsymbol A_j^{(k)}\boldsymbol P^{(k)}$
on the active side and zero on the inactive side. Therefore,
\begin{equation}
\label{eq:projected-activation-matrix-frobenius}
\left\|
\boldsymbol I_+^{(k)}
\boldsymbol A^{(k)}
\boldsymbol P^{(k)}
\right\|_F^2
=
\left\|
\boldsymbol I_-^{(k)}
\boldsymbol A^{(k)}
\boldsymbol P^{(k)}
\right\|_F^2
+
\left\|
\boldsymbol A_j^{(k)}
\boldsymbol P^{(k)}
\right\|^2,
\end{equation}
which 
shows that the active side matrix has a larger Frobenius norm, or
equivalently, greater total squared row energy. 

From
Eq.~\eqref{eq:two-sided-projected-layer-margin-normals}, we have
$\overline{\boldsymbol g}^{(k)}
=
\overline{\boldsymbol g}^{\prime(k)}
+
(\bG_{a,j}^{(k+1)}-\bG_{b,j}^{(k+1)})
\bP^{(k)}(\bA_j^{(k)})^\top$.
Consequently,
\begin{equation}\small
\label{eq:projected-normal-length-difference}
\begin{array}{l}
\left\|
\overline{\boldsymbol g}^{(k)}
\right\|^2
-
\left\|
\overline{\boldsymbol g}^{\prime(k)}
\right\|^2=\\
2
\left(
\bG_{a,j}^{(k+1)}
-
\bG_{b,j}^{(k+1)}
\right)
\left\langle
\overline{\boldsymbol g}^{\prime(k)},
\bP^{(k)}
\left(
\bA_j^{(k)}
\right)^\top
\right\rangle
+
\left(
\bG_{a,j}^{(k+1)}
-
\bG_{b,j}^{(k+1)}
\right)^2
\left\|
\bP^{(k)}
\left(
\bA_j^{(k)}
\right)^\top
\right\|^2.
\end{array}
\end{equation}
 Generally, the last term in
Eq.~\eqref{eq:projected-normal-length-difference} is 
positive. If the angle between
$(\bG_{a,j}^{(k+1)}-\bG_{b,j}^{(k+1)})
\bP^{(k)}(\bA_j^{(k)})^\top$
and
$\overline{\boldsymbol g}^{\prime(k)}$
is no bigger than $\pi/2$, their inner product is nonnegative. The first term
is therefore nonnegative, and
$\|\overline{\boldsymbol g}^{(k)}\|
>
\|\overline{\boldsymbol g}^{\prime(k)}\|$.

If the angle is greater than $\pi/2$, the first term in
Eq.~\eqref{eq:projected-normal-length-difference} is negative,
whereas the last term remains positive. The length ordering is
therefore determined by their relative magnitudes. If the last term is
larger than the absolute value of the first term, then
$\|\overline{\boldsymbol g}^{(k)}\|
>
\|\overline{\boldsymbol g}^{\prime(k)}\|$; otherwise, 
$\|\overline{\boldsymbol g}^{(k)}\|
\leq
\|\overline{\boldsymbol g}^{\prime(k)}\|$.

\paragraph{Statistical interpretation.}
We formalize the preceding intuition using an idealized model.  
\begin{proposition}\label{pro:1}
    We fix
$\bI_+^{(k)}\bA^{(k)}\bP^{(k)}$ and
$\bI_-^{(k)}\bA^{(k)}\bP^{(k)}$, and treat the suffix coefficients
$\bG_a^{(k+1)}-\bG_b^{(k+1)}$ as random. Specifically, we assume that their coordinates are independent zero-mean Gaussian variables with common variance $\sigma^2>0$\footnote{Under Kaiming initialization~\cite{DBLP:conf/iccv/HeZRS15}, all network weights are independent zero-mean Gaussian variables.}. 
We have the expectation
\begin{equation}
\label{eq:expected-projected-normal-length-advantage}
\mathbb E\left[
    \lVert\overline{\boldsymbol g}^{(k)}\rVert^2
    -
    \lVert\overline{\boldsymbol g}^{\prime(k)}\rVert^2
\right]
=
\sigma^2
\lVert\bA_j^{(k)}\bP^{(k)}\rVert^2.
\end{equation}
\end{proposition}
\begin{proof}
Let 
$\boldsymbol g^{(k+1)}=\bG_a^{(k+1)}-\bG_b^{(k+1)}$ with $\boldsymbol g^{(k+1)}_i$ 
denoting its $i$-th coordinate. For a fixed matrix $\boldsymbol M$, let
$\boldsymbol M_i$ denote its $i$-th row. Since
$\boldsymbol g^{(k+1)}\boldsymbol M=\sum_i \bg^{(k+1)}_i\boldsymbol M_i$, bilinearity of
the inner product and linearity of expectation give
\begin{equation}\small
\label{eq:expected-suffix-weighted-norm}
\begin{aligned}
\mathbb E\left[
    \lVert \bg^{(k+1)}\boldsymbol M\rVert^2
\right]
&=
\mathbb E\left[
    \left\langle
        \sum_i \bg^{(k+1)}_i\boldsymbol M_i,
        \sum_j \bg^{(k+1)}_j\boldsymbol M_j
    \right\rangle
\right]\\
&=
\sum_{i,j}
\mathbb E[\bg^{(k+1)}_i \bg^{(k+1)}_j]
\left\langle
    \boldsymbol M_i,\boldsymbol M_j
\right\rangle
=\sigma^2\sum_i
\lVert \boldsymbol M_i\rVert^2=
\sigma^2
\lVert\boldsymbol M\rVert_F^2, 
\end{aligned}
\end{equation}
since independence and zero means give
$\mathbb E[\bg^{(k+1)}_i \bg^{(k+1)}_j]=0$ for $i\neq j$, while
$\mathbb E[(\bg^{(k+1)}_i)^2]=\sigma^2$ for every $i$.

Replacing $\boldsymbol M$ in  Eq.~\eqref{eq:expected-suffix-weighted-norm} by 
$\bI_+^{(k)}\bA^{(k)}\bP^{(k)}$ and
$\bI_-^{(k)}\bA^{(k)}\bP^{(k)}$, and then using
Eq.~\eqref{eq:projected-activation-matrix-frobenius}, gives
Eq.~\eqref{eq:expected-projected-normal-length-advantage}.
\end{proof}

Thus, whenever
$\bA_j^{(k)}\bP^{(k)}\neq\boldsymbol 0$, the active-side
projected normal has a larger expected squared length. 
A positive expected difference alone does not determine how often an
individual comparison is correct. Therefore, we introduce the following single-point
success probability estimation. 
\begin{proposition}[Single-point
success probability]
Follow the same assumption in Pro. \ref{pro:1} and define
$\tau:=
\sum_{i\in\mathbb{S}/\{j\}}
\langle
\bA_i^{(k)}\bP^{(k)},
\bA_j^{(k)}\bP^{(k)}
\rangle^2\geq 0$, then  we have  the single-point success probability,
\begin{equation}
\label{eq:single-point-length-success-probability}
\Pr\left[
    \lVert\overline{\boldsymbol g}^{(k)}\rVert
    >
    \lVert\overline{\boldsymbol g}^{\prime(k)}\rVert
\right]
=
\frac{1}{2}
+
\frac{1}{\pi}
\arcsin\left(
    \frac{
        \lVert\bA_j^{(k)}\bP^{(k)}\rVert^2
    }{
        \sqrt{
            \lVert\bA_j^{(k)}\bP^{(k)}\rVert^4
            +4\tau
        }
    }
\right)
>
\frac{1}{2}.
\end{equation}
\end{proposition}

\begin{proof}
    Define
$\boldsymbol u_i
:=
\bP^{(k)}(\bA_i^{(k)})^\top$.
According to Eq. \eqref{eq:two-sided-projected-layer-margin-normals}, 
$\overline{\boldsymbol g}^{\prime(k)}
=
\sum_{i\in\mathbb{S}/\{j\}}\bg^{(k+1)}_i\boldsymbol u_i$,
$\overline{\boldsymbol g}^{(k)}
=
\sum_{i\in\mathbb{S}/\{j\}}\bg^{(k+1)}_i\boldsymbol u_i
+
\bg^{(k+1)}_j\boldsymbol u_j$, and
$\tau
=
\sum_{i\in\mathbb{S}/\{j\}}
\langle\boldsymbol u_i,\boldsymbol u_j\rangle^2$.
Eq.~\eqref{eq:projected-normal-length-difference} can then be
written as
\begin{equation}
\label{eq:factorized-projected-normal-length-difference}
\lVert\overline{\boldsymbol g}^{(k)}\rVert^2
-
\lVert\overline{\boldsymbol g}^{\prime(k)}\rVert^2
=
\bg^{(k+1)}_j
\left(
2\sum_{i\in\mathbb{S}/\{j\}}
\bg^{(k+1)}_i
\left\langle
\boldsymbol u_i,\boldsymbol u_j
\right\rangle
+
\bg^{(k+1)}_j\lVert\boldsymbol u_j\rVert^2
\right).
\end{equation}
Define
$Z_1:=\bg^{(k+1)}_j$ and
$Z_2:=
2\sum_{i\in\mathbb{S}/\{j\}}
\bg^{(k+1)}_i\langle\boldsymbol u_i,\boldsymbol u_j\rangle
+
\bg^{(k+1)}_j\lVert\boldsymbol u_j\rVert^2$.
The active-side projected normal is therefore longer exactly when
$Z_1Z_2>0$, that is, when $Z_1$ and $Z_2$ have the same sign.

By assumption in Pro. \ref{pro:1}, the coefficients $\bg^{(k+1)}_i$ are independent zero-mean
Gaussian variables and $\bu_i$ are fixed. For any $\alpha,\beta\in\mathbb R$,
$\alpha Z_1+\beta Z_2$ is a linear combination of 
$\bg^{(k+1)}_i$ and is therefore Gaussian. Hence, $(Z_1,Z_2)$ is a jointly
Gaussian pair. Their means are zero by the linearity of expectation.

Since $Z_1=\bg^{(k+1)}_j$, we immediately have
$\operatorname{Var}(Z_1)=\sigma^2$.
For $Z_2$, we have
\begin{equation}
\label{eq:appendix-z2-variance}
\begin{aligned}
\operatorname{Var}(Z_2)
&=
\operatorname{Var}\left(
    2\sum_{i\in\mathbb{S}/\{j\}}
    \bg^{(k+1)}_i
    \left\langle
        \boldsymbol u_i,\boldsymbol u_j
    \right\rangle
    +
    \bg^{(k+1)}_j\lVert\boldsymbol u_j\rVert^2
\right)
\\
&=
4\sum_{i\in\mathbb{S}/\{j\}}
\left\langle
    \boldsymbol u_i,\boldsymbol u_j
\right\rangle^2
\operatorname{Var}(\bg^{(k+1)}_i)
+
\lVert\boldsymbol u_j\rVert^4
\operatorname{Var}(\bg^{(k+1)}_j)
=
\sigma^2
\left(
    4\tau+\lVert\boldsymbol u_j\rVert^4
\right).
\end{aligned}
\end{equation}

Since $Z_1$ and $Z_2$ have zero mean, their covariance is
$\operatorname{Cov}(Z_1,Z_2)=\mathbb E[Z_1Z_2]$.
For  $i\in\mathbb{S}/\{j\}$, the independence gives
$\mathbb E[\bg^{(k+1)}_j\bg^{(k+1)}_i]=\mathbb E[\bg^{(k+1)}_j]\mathbb E[\bg^{(k+1)}_i]=0$. Therefore,
\begin{equation}\small
\label{eq:appendix-z1-z2-covariance}
\begin{aligned}
\operatorname{Cov}(Z_1,Z_2)
&=
\mathbb E\left[
\bg^{(k+1)}_j
\left(
2\sum_{i\in\mathbb{S}/\{j\}}
\bg^{(k+1)}_i
\left\langle
\boldsymbol u_i,\boldsymbol u_j
\right\rangle
+
\bg^{(k+1)}_j\lVert\boldsymbol u_j\rVert^2
\right)
\right]\\
&=
2\sum_{i\in\mathbb{S}/\{j\}}
\left\langle
\boldsymbol u_i,\boldsymbol u_j
\right\rangle
\mathbb E[\bg^{(k+1)}_j\bg^{(k+1)}_i]
+
\lVert\boldsymbol u_j\rVert^2
\mathbb E[(\bg^{(k+1)}_j)^2]
=
\sigma^2
\lVert\boldsymbol u_j\rVert^2.
\end{aligned}
\end{equation} 
Using
$\operatorname{Var}(Z_1)=\sigma^2$,
$\operatorname{Var}(Z_2)
=
\sigma^2(4\tau+\lVert\boldsymbol u_j\rVert^4)$, and
$\operatorname{Cov}(Z_1,Z_2)
=
\sigma^2\lVert\boldsymbol u_j\rVert^2$,
the correlation coefficient between $Z_1$ and $Z_2$ is
\begin{equation}
\label{eq:appendix-z1-z2-correlation}
\operatorname{Corr}(Z_1,Z_2)
=
\frac{
\operatorname{Cov}(Z_1,Z_2)
}{
\sqrt{
\operatorname{Var}(Z_1)
\operatorname{Var}(Z_2)
}
}
=
\frac{
\lVert\boldsymbol u_j\rVert^2
}{
\sqrt{
\lVert\boldsymbol u_j\rVert^4+4\tau
}
}.
\end{equation}
Since $\boldsymbol u_j\neq\boldsymbol 0$,  $\operatorname{Corr}(Z_1,Z_2)>0$. 
Thus, $Z_1$ and $Z_2$ are positively correlated.

Dividing $Z_1$ and $Z_2$ by their positive standard deviations does
not change their signs or their correlation coefficient. The resulting
variables form a standard jointly Gaussian pair. For such a pair, the
standard Gaussian quadrant identity gives
$\Pr[Z_1>0,Z_2>0]
=
1/4+
\arcsin(\operatorname{Corr}(Z_1,Z_2))/(2\pi)$.
Moreover, the zero-mean jointly Gaussian distribution is centrally
symmetric, so
$\Pr[Z_1<0,Z_2<0]=\Pr[Z_1>0,Z_2>0]$. 
By
Eq.~\eqref{eq:factorized-projected-normal-length-difference},
the active-side projected normal is longer exactly when $Z_1Z_2>0$.
Therefore,
\begin{equation}
\label{eq:appendix-single-point-success-probability}
\begin{aligned}
\Pr\left[
    \lVert\overline{\boldsymbol g}^{(k)}\rVert
    >
    \lVert\overline{\boldsymbol g}^{\prime(k)}\rVert
\right]
=
\Pr[Z_1Z_2>0]&=
\frac{1}{2}
+
\frac{1}{\pi}
\arcsin\left(
    \operatorname{Corr}(Z_1,Z_2)
\right)
\\
&=
\frac{1}{2}
+
\frac{1}{\pi}
\arcsin\left(
    \frac{
        \lVert\boldsymbol u_j\rVert^2
    }{
        \sqrt{
            \lVert\boldsymbol u_j\rVert^4+4\tau
        }
    }
\right).
\end{aligned}
\end{equation}
Since $\operatorname{Corr}(Z_1,Z_2)>0$, the probability in
Eq.~\eqref{eq:appendix-single-point-success-probability} is
strictly greater than $1/2$. Finally, substituting
$\boldsymbol u_j=\bP^{(k)}(\bA_j^{(k)})^\top$
proves
Eq.~\eqref{eq:single-point-length-success-probability}.
\end{proof}

\begin{proposition}
    The $\tau:=
\sum_{i\in\mathbb{S}/\{j\}}
\langle
\bA_i^{(k)}\bP^{(k)},
\bA_j^{(k)}\bP^{(k)}
\rangle^2$  is the sum of the squared inner products between the projected  weight of the target neuron (neuron $j$) and the projected weights of the other active neurons (neuron $i\in\mathbb{S}/\{j\}$). 
It therefore measures their total alignment with the target direction $\bA_j^{(k)}\bP^{(k)}$. A larger alignment $\tau$ generally reduces the single-point success probability. 
\end{proposition}



\subsection{\em{Normal Alignment} in Hard-Label Setting}
\label{subsubsec:alignment-rule}

The preceding analysis gives the length comparison: according to Eq. \eqref{eq:single-point-length-success-probability}, the side with the longer
$\overline{\boldsymbol g}^{(k)}$ is predicted to be the target active
side with a higher probability. However, the hard-label queries recover only the unit projected
normals
$\overline{\boldsymbol n}^{(k)}$ and
$\overline{\boldsymbol n}^{\prime(k)}$ as stated in  the last paragraph of Sect. \ref{subsec:recovering-projected-unit-normals},
thereby losing the lengths of
$\overline{\boldsymbol g}^{(k)}$ and
$\overline{\boldsymbol g}^{\prime(k)}$.
Consequently, the length comparison cannot be applied directly. We therefore
compare the absolute alignments of the two unit projected normals with
the recovered target weight $\widehat{\bA}_j^{(k)}$.

\begin{proposition}[Equivalence of the length and alignment comparisons]
\label{prop:length-alignment-equivalence}
Assume 
$\overline{\boldsymbol g}^{(k)}\neq\boldsymbol 0$,
$\overline{\boldsymbol g}^{\prime(k)}\neq\boldsymbol 0$,
$\bP^{(k)}(\widehat{\bA}_j^{(k)})^\top\neq\boldsymbol 0$, and
$\bP^{(k)}(\widehat{\bA}_j^{(k)})^\top
\nparallel
\overline{\boldsymbol g}^{\prime(k)}$.
Then
\begin{equation}\label{eq:NA_equivalence}
   \left\|\overline{\boldsymbol g}^{(k)}\right\|
>
\left\|\overline{\boldsymbol g}^{\prime(k)}\right\|
\quad\Longleftrightarrow\quad
\left|
\left\langle
\overline{\boldsymbol n}^{(k)},
\left(\widehat{\bA}_j^{(k)}\right)^\top
\right\rangle
\right|
>
\left|
\left\langle
\overline{\boldsymbol n}^{\prime(k)},
\left(\widehat{\bA}_j^{(k)}\right)^\top
\right\rangle
\right|. 
\end{equation}
\end{proposition}

\begin{proof}
According to Eq. \eqref{eq:n_bar}, we have
$\bP^{(k)}\overline{\boldsymbol n}^{(k)}
=
\overline{\boldsymbol n}^{(k)}$
and
$\bP^{(k)}\overline{\boldsymbol n}^{\prime(k)}
=
\overline{\boldsymbol n}^{\prime(k)}$.
Define
$\boldsymbol m
:=
\bP^{(k)}(\widehat{\bA}_j^{(k)})^\top/
\|\bP^{(k)}(\widehat{\bA}_j^{(k)})^\top\|$.
Using these identities and the symmetry of $\bP^{(k)}$, we obtain
\begin{equation}
\label{eq:appendix-alignment-projection}
\resizebox{\linewidth}{!}{$
\begin{aligned}
\left\langle
\overline{\boldsymbol n}^{(k)},
\left(\widehat{\bA}_j^{(k)}\right)^\top
\right\rangle
&=
\left\langle
\bP^{(k)}\overline{\boldsymbol n}^{(k)},
\left(\widehat{\bA}_j^{(k)}\right)^\top
\right\rangle
=
\left\langle
\overline{\boldsymbol n}^{(k)},
\bP^{(k)}
\left(\widehat{\bA}_j^{(k)}\right)^\top
\right\rangle
=
\left\|
\bP^{(k)}
\left(\widehat{\bA}_j^{(k)}\right)^\top
\right\|
\left\langle
\overline{\boldsymbol n}^{(k)},
\boldsymbol m
\right\rangle,
\\
\left\langle
\overline{\boldsymbol n}^{\prime(k)},
\left(\widehat{\bA}_j^{(k)}\right)^\top
\right\rangle
&=
\left\langle
\bP^{(k)}\overline{\boldsymbol n}^{\prime(k)},
\left(\widehat{\bA}_j^{(k)}\right)^\top
\right\rangle
=
\left\langle
\overline{\boldsymbol n}^{\prime(k)},
\bP^{(k)}
\left(\widehat{\bA}_j^{(k)}\right)^\top
\right\rangle
=
\left\|
\bP^{(k)}
\left(\widehat{\bA}_j^{(k)}\right)^\top
\right\|
\left\langle
\overline{\boldsymbol n}^{\prime(k)},
\boldsymbol m
\right\rangle.
\end{aligned}
$}
\end{equation}
The common factor
$\|\bP^{(k)}(\widehat{\bA}_j^{(k)})^\top\|$
is positive and therefore does not affect the ordering of the absolute
inner products.

Since $\widehat{\bA}_j^{(k)}$ and $\bA_j^{(k)}$ differ only by a
nonzero scalar,
Eq.~\eqref{eq:two-sided-projected-layer-margin-normals} shows that
$\overline{\boldsymbol g}^{(k)}
-
\overline{\boldsymbol g}^{\prime(k)}$
is parallel to $\boldsymbol m$. Hence, there exist
$\alpha,\beta\in\mathbb R$ and a vector
$\boldsymbol m_\perp\perp\boldsymbol m$ such that
$\overline{\boldsymbol g}^{\prime(k)}
=
\alpha\boldsymbol m+\boldsymbol m_\perp$
and
$\overline{\boldsymbol g}^{(k)}
=
\beta\boldsymbol m+\boldsymbol m_\perp$.
Thus, 
they share the same orthogonal component $\boldsymbol m_\perp$. 
Since $\boldsymbol m$ is a unit vector and
$\boldsymbol m_\perp\perp\boldsymbol m$, we have
$\|\overline{\boldsymbol g}^{\prime(k)}\|^2
=
\alpha^2+\|\boldsymbol m_\perp\|^2$
and
$\|\overline{\boldsymbol g}^{(k)}\|^2
=
\beta^2+\|\boldsymbol m_\perp\|^2$.
Using
$\overline{\boldsymbol n}^{(k)}
=
\overline{\boldsymbol g}^{(k)}
/
\|\overline{\boldsymbol g}^{(k)}\|$
and
$\overline{\boldsymbol n}^{\prime(k)}
=
\overline{\boldsymbol g}^{\prime(k)}
/
\|\overline{\boldsymbol g}^{\prime(k)}\|$,
we obtain
$|\langle
\overline{\boldsymbol n}^{(k)},
\boldsymbol m
\rangle|^2
=
\beta^2/
(\beta^2+\|\boldsymbol m_\perp\|^2)$
and
$|\langle
\overline{\boldsymbol n}^{\prime(k)},
\boldsymbol m
\rangle|^2
=
\alpha^2/
(\alpha^2+\|\boldsymbol m_\perp\|^2)$.
Consequently,
\begin{equation}
\begin{array}{ll}
&\left|
\left\langle
\overline{\boldsymbol n}^{(k)},
\boldsymbol m
\right\rangle
\right|^2
-
\left|
\left\langle
\overline{\boldsymbol n}^{\prime(k)},
\boldsymbol m
\right\rangle
\right|^2\\
&=
\frac{
    \beta^2
}{
    \beta^2+\|\boldsymbol m_\perp\|^2
}
-
\frac{
    \alpha^2
}{
    \alpha^2+\|\boldsymbol m_\perp\|^2
}
=
\frac{
    \beta^2
    \left(
        \alpha^2+\|\boldsymbol m_\perp\|^2
    \right)
    -
    \alpha^2
    \left(
        \beta^2+\|\boldsymbol m_\perp\|^2
    \right)
}{
    \left(
        \beta^2+\|\boldsymbol m_\perp\|^2
    \right)
    \left(
        \alpha^2+\|\boldsymbol m_\perp\|^2
    \right)
}\\
&=
\frac{
    \|\boldsymbol m_\perp\|^2
    \left(
        \beta^2-\alpha^2
    \right)
}{
    \left(
        \beta^2+\|\boldsymbol m_\perp\|^2
    \right)
    \left(
        \alpha^2+\|\boldsymbol m_\perp\|^2
    \right)
}
=
\frac{
    \|\boldsymbol m_\perp\|^2
}{
    \|\overline{\boldsymbol g}^{(k)}\|^2
    \|\overline{\boldsymbol g}^{\prime(k)}\|^2
}
\left(
    \|\overline{\boldsymbol g}^{(k)}\|^2
    -
    \|\overline{\boldsymbol g}^{\prime(k)}\|^2
\right). 
\end{array}
\end{equation} 
With the assumptions $\overline{\boldsymbol g}^{(k)}\neq\boldsymbol 0$,
$\overline{\boldsymbol g}^{\prime(k)}\neq\boldsymbol 0$,
$\bP^{(k)}(\widehat{\bA}_j^{(k)})^\top\neq\boldsymbol 0$, and
$\bP^{(k)}(\widehat{\bA}_j^{(k)})^\top
\nparallel
\overline{\boldsymbol g}^{\prime(k)}$, 
the factor
$\|\boldsymbol m_\perp\|^2/
(\|\overline{\boldsymbol g}^{(k)}\|^2
\|\overline{\boldsymbol g}^{\prime(k)}\|^2)$
is strictly positive. Hence, 
$\|\overline{\boldsymbol g}^{(k)}\|
>
\|\overline{\boldsymbol g}^{\prime(k)}\|$
if and only if
$|\langle\overline{\boldsymbol n}^{(k)},\boldsymbol m\rangle|
>
|\langle\overline{\boldsymbol n}^{\prime(k)},\boldsymbol m\rangle|$.
By
Eq.~\eqref{eq:appendix-alignment-projection},
the latter is equivalent to
$|\langle
\overline{\boldsymbol n}^{(k)},
(\widehat{\bA}_j^{(k)})^\top
\rangle|
>
|\langle
\overline{\boldsymbol n}^{\prime(k)},
(\widehat{\bA}_j^{(k)})^\top
\rangle|$,
which proves
Pro.~\ref{prop:length-alignment-equivalence}.
\end{proof}

According to Pro.~\ref{prop:length-alignment-equivalence}, 
the alignment comparison has
the same single-point success probability given in
Eq.~\eqref{eq:single-point-length-success-probability}. 
Based on this comparison, we introduce the sign recovery method
\emph{Normal Alignment} for hard-label networks. 
\begin{proposition}[\em Normal Alignment]
At each dual point,
the side whose unit projected decision-facet normal $\overline{\boldsymbol n}^{(k)}$ has the larger
absolute inner product with
$(\widehat{\bA}_j^{(k)})^\top$
is predicted to be the target active side. If this prediction agrees
with the active side indicated by the recovered signature, the dual point votes to retain the  sign of $\widehat{\bA}_j^{(k)}$; otherwise, it votes to reverse its sign. 
\end{proposition}

{\em Normal Alignment} repeats this comparison at multiple dual points and aggregates the resulting votes. The majority vote determines whether the sign of the recovered signature is retained or reversed. The fraction of valid votes supporting this decision is used as its confidence level. For a target neuron, let $n_+$ and $n_-$ denote the numbers of votes for
retaining and reversing the recovered signature, respectively, and let
$n_{\mathrm{valid}}:=n_++n_-$. For $n_{\mathrm{valid}}>0$, define the confidence level as
$\alpha:=\max\{n_+,n_-\}/n_{\mathrm{valid}}$. Given a confidence threshold
$\alpha_0\in(1/2,1]$, the {\em Normal Alignment} retains the recovered signature
if $n_+>n_-$ and $\alpha\geq\alpha_0$, and reverses it if $n_->n_+$ and
$\alpha\geq\alpha_0$. Otherwise, the available votes are insufficient
to determine the sign, which remains unresolved.

\subsubsection{White-Box Validation of the {\em Normal Alignment}.}
\label{subsubsec:normal-alignment-validation}

We experimentally validate the  {\em Normal Alignment} in the white-box setting on the CIFAR-10 DNNs 
with architectures \mbox{$192$-$d$$\times$3-$10$} for
$d\in\{32,64,128,256\}$.  As shown in Fig.~\ref{fig:normal-alignment-dual-point-budget} in {\sf Supp.}~\ref{supp:supporting-experimental-results}, $n_{\mathrm{attempt}}=200$ provides high sign recovery accuracy.  
Table~\ref{tab:normal-alignment-whitebox-validation} summarizes the layer-wise results.  
For each dual
point, the projected normal length comparison in
Eq.~\eqref{eq:projected-normal-length-difference} is quantified by $\frac{\lVert\overline{\boldsymbol g}^{(k)}\rVert}
     {\lVert\overline{\boldsymbol g}^{\prime(k)}\rVert}$. The corresponding column 
reports the median of this ratio over all evaluated dual points in each
layer.  For each neuron in a layer, Eq.~\eqref{eq:single-point-length-success-probability}
is used to compute a theoretical success probability at each of its
dual points and  ``$p_{\mathrm{th}}$'' is the mean of these
probabilities over all evaluated dual points in the layer.
Correspondingly, $p_{\mathrm{obs}}$ is the proportion of correct
single-point votes among all evaluated dual points in that layer.  The column of ``Agreement'' reports the percentage of evaluated dual points in each layer satisfying the equivalence in Eq. \eqref{eq:NA_equivalence} of Pro. \ref{prop:length-alignment-equivalence}. The column of ``Signs recovered'' 
gives the number of correctly recovered neuron signs in each layer after aggregating
$200$ votes per neuron.

\begin{table}
\centering
\caption{Layer-wise white-box validation of  {\em Normal Alignment} on CIFAR-10 DNNs \mbox{$192$-$d$-$d$-$d$-$10$} using $n_{\mathrm{attempt}}=200$ dual points per neuron.}
\label{tab:normal-alignment-whitebox-validation}
\scriptsize
\setlength{\tabcolsep}{4.2pt}
\renewcommand{\arraystretch}{1.08}
\begin{tabular}{@{}cccccc@{}}
\toprule
$d$
& Layer $k$
& $\lVert\overline{\boldsymbol g}^{(k)}\rVert/
   \lVert\overline{\boldsymbol g}^{\prime(k)}\rVert$
& $p_{\mathrm{th}}/p_{\mathrm{obs}}$ (\%)
& Agreement (\%)
& Signs recovered
\\
\midrule
\multirow{3}{*}{$32$}
& 1 & 1.014 & 74.94/76.61 & 100.00 & 32/32 \\
& 2 & 1.010 & 64.31/66.92 & 100.00 & 30/32 \\
& 3 & 1.016 & 66.86/69.56 & 100.00 & 31/32 \\
\midrule
\multirow{3}{*}{$64$}
& 1 & 1.008 & 71.04/74.68 & 100.00 & 64/64 \\
& 2 & 1.007 & 65.12/69.12 & 100.00 & 64/64 \\
& 3 & 1.012 & 66.25/71.84 & 100.00 & 62/64 \\
\midrule
\multirow{3}{*}{$128$}
& 1 & 1.004 & 68.04/73.15 & 100.00 & 128/128 \\
& 2 & 1.004 & 65.62/68.79 & 100.00 & 128/128 \\
& 3 & 1.004 & 65.40/68.82 & 100.00 & 126/128 \\
\midrule
\multirow{3}{*}{$256$}
& 1 & 1.002 & 64.99/70.81 & 100.00 & 256/256 \\
& 2 & 1.002 & 66.00/68.95 & 100.00 & 255/256 \\
& 3 & 1.002 & 64.41/68.73 & 100.00 & 241/256 \\
\bottomrule
\end{tabular}
\end{table}

As shown in
Table~\ref{tab:normal-alignment-whitebox-validation}, the median of the ratio$\frac{\lVert\overline{\boldsymbol g}^{(k)}\rVert}
     {\lVert\overline{\boldsymbol g}^{\prime(k)}\rVert}$ is greater than one in all twelve hidden
layers. For all the evaluated
dual points of the four DNNs, the mean single-point success probabilities are
$p_{\mathrm{th}}=66.01\%$ theoretically and
$p_{\mathrm{obs}}=70.12\%$ empirically.
Thus, the theoretical model captures the advantage over random
guessing, although it underestimates its magnitude. Consistent with
Pro.~\ref{prop:length-alignment-equivalence}, the length and
alignment rules agree (Eq. \eqref{eq:NA_equivalence}  is satisfied) on every evaluated dual point. 
After vote aggregation, {\em Normal Alignment} correctly recovers $1417$ of
the $1440$ neuron signs, giving an aggregate recovery accuracy of
$\frac{1417}{1440}\times100\%=98.40\%$. In particular, all signs in the
first hidden layer are recovered correctly for all four models.
Meanwhile, the second and third hidden layers contain $20$ incorrectly
recovered signs and $3$ tied outcomes.



%% file: sections/combineSOE.tex

\section{{\em eSOE+Alignment}: Combining {\em Normal Alignment} with Hard-Label {\em SOE}}
\label{sec:soe-align}

At NeurIPS 2024, Foerster {\em et al.}~\cite{DBLP:conf/nips/FoersterMSH24} empirically observed that many
neuron signs recovered by {\em Neuron Wiggle} \cite{DBLP:conf/eurocrypt/CanalesMartinezCHRSS24} remained at low confidence and that collecting additional critical points did not improve their confidence.
Therefore, they performed an exhaustive search on these low-confidence signs, which led to a significant increase in the number of model queries and runtime, even turning the so-called polynomial-time attack into an exponential-time  attack.
As a probabilistic voting method, the {\em Normal Alignment} may also leave some neuron signs undetermined when their voting confidence is insufficient. 
In contrast, Hard-label {\em SOE} \cite{DBLP:conf/latincrypt/CanalesMartinezS25} can deterministically recover the activation states of all neurons in layer $k$, but it requires $\operatorname{rank}(\widehat{\bA}^{(k)}\bF^{(k-1)})=d^{(k+1)}$ to solve the linear system in Eq.~\eqref{eq:hard-label SOE}. 
Inspired by the raw-output {\em SOE}+{\em Wiggle} \cite{DBLP:conf/eurocrypt/LiuSELBP26}, we proposed hard-label {\em eSOE}+{\em Alignment} for high-confidence sign recovery.



Recall from Sect.~\ref{subsec:SOE and its extensions}, {\em SOE} + {\em Wiggle} contains two main parts: removing inactive neurons identified with high-confidence level in Eq. \eqref{eq:reduced-raw-output-SOE} and extending the system from multiple points in Eq. \eqref{eq:raw-output-SOE-extension}.
In S1 access, removing inactive neurons remains straightforward: at a selected transition point, the signs recovered by {\em Normal Alignment} with high confidence are used to identify inactive neurons, and their corresponding zero entries are removed from the hard-label {\em SOE} Eq. \eqref{eq:hard-label SOE} in Sect.~\ref{subsec:SOE and its extensions}. 
The system extension, however, cannot be applied directly. 


\subsubsection{Finding Compatible Transition Points.}
\label{subsubsec:soe-align-extension-points}
According to Eq.~\eqref{eq:hard-label SOE}, the unknown vector is $\left(
    (\bG_a^{(k+1)}-\bG_b^{(k+1)})
    \bI^{(k)}
    \bS^{(k)}
\right)^\top$. To extend the linear system in Eq.~\eqref{eq:hard-label SOE}, the selected points must share the same unknowns. Therefore, two conditions should be simultaneously satisfied: 
\begin{itemize}
    \item First, the selected points must share the activation states in layer $k$ and all subsequent layers; otherwise, their corresponding hard-label {\em SOE} systems in Eq.~\eqref{eq:hard-label SOE} have different unknown vectors;
    \item Second, the points should remain on the decision boundary between the same two classes, {\em i.e.}, $\mathbb D_{ab}$.
\end{itemize}
We can only walk along the decision boundary to ensure that no future neurons have toggled, and check the output label from slightly perturbing $\mathcal{F}_{\theta}(\bx_t+\bdelta^{(1)})$ and $\mathcal{F}_{\theta}(\bx_t-\bdelta^{(1)})$ to keep $\bx_t \in \mathbb D_{ab}$.
We use the idea for locating dual points introduced in
\cite{DBLP:conf/eurocrypt/CarliniCHRS25} to find compatible transition
points on the decision boundary between fixed two classes, specifically, 
\begin{itemize}
    \item {\em Step 1}: From $\bx_1 \in \mathbb{D}_{ab} \cap \mathbb L_{\bx_1}$, the attacker makes a random excursion and uses hard-label binary search to relocate another point $\bx_1' \in \mathbb D_{ab} \cap \mathbb L_{\bx_1}$. Then the difference $\bx_1'-\bx_1$ determines a direction along $\mathbb D_{ab} \cap \mathbb L_{\bx_1}$. Following the direction of $\bx_1'-\bx_1$, the attacker walks until the decision boundary bends at a dual point $\bx_2$ (suppose that $\bx_1 \in \mathbb D_{ab} \cap \mathbb L_{\bx_2}$).
    \item {\em Step 2}: Since the parameters of layers $1,\ldots,k-1$ have already been recovered, the attacker can evaluate the pre-activation values of the neurons in these layers at $\bx_2$. If a neuron in a preceding layer is zero before ReLU, the bend is attributed to that layer; the attacker then makes a random excursion and uses binary search to locate another $\bx_2'$ on the adjacent decision facet $\mathbb D_{ab} \cap \mathbb L'_{\bx_2}$. Otherwise, the bend may be caused by a neuron in layer $k$ or a subsequent layer, so the attacker terminates the current search path.
    \item {\em Step 3}: If the bend is attributed to a neuron in a
    preceding layer,the attacker then collects $\bx_2'$ for {\em SOE} extension and continues to find the next decision boundary bend along the direction of $\bx'_2-\bx_2$. 
\end{itemize}
Repeating this procedure from the starting point $\bx_1$ yields sufficient transition points for the following extension.

\subsubsection{Scale-Free Hard-label {\em SOE} Extension.}
\label{subsubsec:scale-free-soe-extension}

Let $\bx_1,\ldots,\bx_T$ be the collected transition points. Since these points share the same activation matrices in layer $k$ and all subsequent layers, their $\bI^{(k)}$ and $\bG^{(k+1)}$ are identical. 
Let $\boldsymbol g^{(1)}_t$ be the normal of the decision facet  $\mathbb D_{ab} \cap \mathbb L_{\bx_t}\subset \mathbb{R}^{d^{(1)}}$, 
and $\boldsymbol n^{(1)}_t = \boldsymbol g^{(1)}_t/\|\boldsymbol g^{(1)}_t\|$.  The decision-facet normal is often already available from signature recovery and can therefore be reused here, avoiding the additional oracle queries needed to search for perturbations and construct the equations in Eq.~\eqref{eq:hard-label SOE}. According to Eq. \eqref{eq:g and gk} in Sect.~\ref{subsec:recovering-projected-unit-normals}, 
$\bg^{(1)}_t = (\bA^{(k)} \bF^{k-1}_{\bx_t})^\top \bI^{(k)} (\bG^{(k+1)}_a-\bG^{(k+1)}_b)^\top$. With the recovered signature $\widehat{\bA}^{(k)}$ and prefix map $\bF_{\bx_t}^{(k-1)}$, we have
\begin{equation}
\label{eq:Hard-label SOE normal}
    \boldsymbol n^{(1)}_t
=
\left(
    \widehat{\bA}^{(k)}
    \bF_{\bx_t}^{(k-1)}
\right)^\top
\frac{\bS^{(k)}}{\|\boldsymbol g^{(1)}_t\|}
\left[
    \bI^{(k)}
    \left(
        \bG_a^{(k+1)}
        -
        \bG_b^{(k+1)}
    \right)^\top
\right].
\end{equation}

Although $\bS^{(k)}\bI^{(k)} (\bG_a^{(k+1)}-\bG_b^{(k+1)})^\top$ is common to all selected points, the $1/\|\boldsymbol g^{(1)}_t\|$ generally differs across them. Therefore, these hard-label {\em SOE} systems cannot be combined directly. Since $\boldsymbol n^{(1)}_t$ is a unit vector, 
$\left(
\mathbf{Id}_{d^{(1)}}-
\bn^{(1)}_t\bn^{(1)\top}_t\right)\bn^{(1)}_t=\boldsymbol 0$, where $\mathbf{Id}_{d^{(1)}}$ is the $d^{(1)} \times d^{(1)}$ identity matrix. Multiplying both sides of Eq.~\eqref{eq:Hard-label SOE normal} by
$\mathbf{Id}_{d^{(1)}}-
\bn^{(1)}_t\bn^{(1)\top}_t$ to remove $1/\|\boldsymbol g^{(1)}_t\|$, and it gives
\begin{equation}
\label{eq:projected-hard-label-soe}
\left(
\mathbf{Id}_{d^{(1)}}
-
\boldsymbol n^{(1)}_t\boldsymbol n^{(1)\top}_t
\right)
\left(
\widehat{\bA}^{(k)}
\bF_{\bx_t}^{(k-1)}
\right)^\top
\bS^{(k)}
\left[
\bI^{(k)}
\left(
\bG_a^{(k+1)}
-
\bG_b^{(k+1)}
\right)^\top
\right]
=
\boldsymbol 0.
\end{equation}
Now, the unknown vector in Eq.~\eqref{eq:projected-hard-label-soe} is the same for all selected transition points.

Again, let $\mathbb K$ contain the neurons that {\em Normal Alignment} identifies as inactive in layer $k$ at $\bx_t$ with high confidence, and
$\mathbb U:=[d^{(k+1)}]\setminus\mathbb K$ contain the remaining neurons. For a vector $\boldsymbol v$, $[\boldsymbol v]_{\mathbb U}$ denotes the subvector indexed by $\mathbb U$.
Restricting Eq.~\eqref{eq:projected-hard-label-soe} to $\mathbb U$ and
combining the equations from all collected transition points, we obtain
\begin{equation}
\label{eq:soe-align-stacked-system}
\begin{bmatrix}
[\left(
    \mathbf{Id}_{d^{(1)}}
    -
    \boldsymbol n^{(1)}_1\boldsymbol n^{(1)\top}_1
\right)
\left(
    \widehat{\bA}^{(k)}
    \bF_{\bx_1}^{(k-1)}
\right)^\top]_{\mathbb U}
\\
[\left(
    \mathbf{Id}_{d^{(1)}}
    -
    \boldsymbol n^{(1)}_2\boldsymbol n^{(1)\top}_2
\right)
\left(
    \widehat{\bA}^{(k)}
    \bF_{\bx_2}^{(k-1)}
\right)^\top]_{\mathbb U}
\\
\vdots
\\
[\left(
    \mathbf{Id}_{d^{(1)}}
    -
    \boldsymbol n^{(1)}_T\boldsymbol n^{(1)\top}_T
\right)
\left(
    \widehat{\bA}^{(k)}
    \bF_{\bx_T}^{(k-1)}
\right)^\top]_{\mathbb U}
\end{bmatrix}
\left[
    \bS^{(k)}
    \bI^{(k)}
    \left(
        \bG_a^{(k+1)}
        -
        \bG_b^{(k+1)}
    \right)^\top
\right]_{\mathbb U}
=
\boldsymbol 0.
\end{equation}

\begin{proposition}[{\em eSOE+Alignment}]
\label{prop:esoe-alignment}
Assume every neuron in $\mathbb K$ is inactive at the selected transition points and  the entries of $\left[(\bG_a^{(k+1)}-\bG_b^{(k+1)})^\top\right]_{\mathbb U}$ corresponding to active neurons are nonzero. If the coefficient matrix in Eq.~\eqref{eq:soe-align-stacked-system} has rank $|\mathbb U|-1$, then its nonzero solution is uniquely determined up to a scalar multiple. Then, 
{eSOE+Alignment} recovers all the neuron signs in layer $k$ in polynomial time.
\end{proposition}


The correctness of the sign assignment returned by {\em eSOE+Alignment} can be assessed by attempting signature recovery for the next layer, as proposed by \cite{DBLP:conf/nips/FoersterMSH24}. If the next-layer signatures cannot be recovered, {\em eSOE} may have incorrectly eliminated a column corresponding to a neuron that is actually active at the selected transition points, making the solution unreliable.We therefore discard the signs returned by {\em eSOE+Alignment} and instead use the signs predicted by {\em Normal Alignment}. We then apply the exhaustive-search strategy of \cite{DBLP:conf/nips/FoersterMSH24} to enumerate candidate assignments for the low-confidence signs until the next-layer signatures are successfully recovered. We handle {\em eSOE+Toggle} analogously: when its reduced {\em SOE} solution is unreliable, we apply the same exhaustive-search strategy to the low-confidence signs predicted by {\em Future Toggle}.

We validate {\em eSOE+Alignment} on the CIFAR-10 model with architecture \mbox{$3072$-$256{\times 3}$-$64$-$10$}. Following the proof-of-concept setting in \cite{DBLP:conf/eurocrypt/CarliniCHRS25}, we use exact decision-facet normals and set $n_{\mathrm{attempt}}=200$. For every hidden layer,
the projected coefficient matrix satisfies the rank condition in Proposition~\ref{prop:esoe-alignment}, and {\em eSOE+Alignment} correctly recovers all $832$ signs, as reported in Table~\ref{tab:cifar-sign-strategies} in {\sf Supp.}~\ref{supp:supporting-experimental-results}. The same result is observed in Table~\ref{tab:intro-combined-sign-recovery}, where {\em eSOE+Alignment} recovers all $512$ and $320$ signs, respectively. Thus, no fallback exhaustive search is required in any of the evaluated settings. For the \mbox{$3072$-$256{\times 3}$-$64$-$10$} model, complete sign recovery takes about 1h43m on our 8-core CPU without neuron-level parallelism. For reference, Carlini {\em et al.}~ \cite{DBLP:conf/eurocrypt/CarliniCHRS25} reported an estimated runtime of 8.5 hours for complete sign recovery using {\em Future Toggle} on a 256-core server, with 64 neurons processed in parallel.

%% file: sections/experiment.tex
\section{Experiments}
\label{sec:experiments}



We evaluate sign recovery on two trained ReLU DNNs: a CIFAR-10 model with architecture \mbox{$192$-$64{\times 8}$-$10$} and an MNIST model with architecture \mbox{$64$-$96{\times 3}$-$32$-$10$}, containing $42{,}122$ and $28{,}298$ parameters, respectively.\footnote{CIFAR-10 and MNIST images are resized and flattened to
$8\times8\times3=192$ and $8\times8=64$ dimensions, respectively.} Both networks use ReLU activations and Kaiming normal initialization~\cite{DBLP:conf/iccv/HeZRS15}.  For the target layer $k$, the sign recoveries are evaluated assuming that the preceding layers have been recovered and the layer $k$'s neuron signatures are known up to sign. 

Following the experimental convention of \cite{DBLP:conf/eurocrypt/CarliniCHRS25}, we assume that the dual points have been precomputed because this step is shared by signature recovery and sign recovery. 
We use $200$ attempted dual points per neuron for {\em Normal Alignment} and {\em eSOE+Alignment}, and evaluate
{\em Future Toggle} and {\em eSOE+Toggle} with budgets of both $200$ and $1000$. For CIFAR-10, we follow the proof-of-concept setting of ~\cite{DBLP:conf/eurocrypt/CarliniCHRS25}. Decision-facet normals are computed directly from model parameters. {\em Future Toggle} additionally uses white-box information to walk along the decision boundary and handle non-future toggles. For MNIST, we recover the unit decision-facet normals using the procedure described in Sect.~\ref{subsubsec:recovering-input-space-unit-normal}, and {\em Future Toggle} discards a dual point if the boundary walk first encounters a non-future toggle. 
Tables~\ref{tab:highest-error-confidence-rank}, \ref{tab:cifar-layer-comparison},  and~\ref{tab:mnist-layer-comparison} in
{\sf Supp.}~\ref{supp:supporting-experimental-results} summarize the results across all hidden layers. {\em eSOE+Alignment} correctly recovers all $512$ and $320$ signs in the two models, respectively. 

%% file: sections/appendixExperimentalResults.tex
\begin{figure}[!htbp]
\centering
\includegraphics[width=0.7\linewidth]
{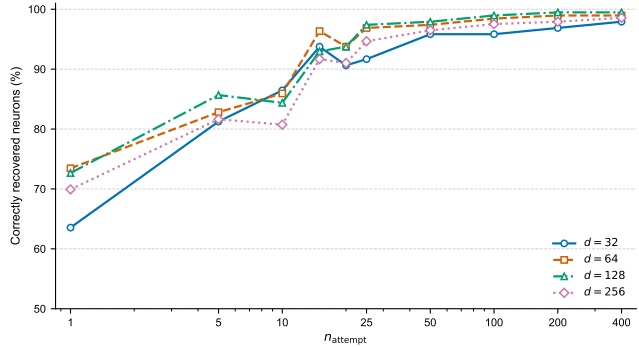}
\caption{Sign recovery accuracy of {\em Normal Alignment}.  
Each curve reports the percentage of correctly recovered signs across all three hidden layers after aggregating the  $n_{\mathrm{attempt}}$ votes for each neuron.}
\label{fig:normal-alignment-dual-point-budget}
\end{figure}


\begin{table}[!htbp]
\centering
\caption{Sign recovery results of {\em Normal Alignment} and
{\em eSOE+Alignment} on the CIFAR-10 model with architecture
\mbox{$3072$-$256{\times 3}$-$64$-$10$}.}
\label{tab:cifar-sign-strategies}

\begingroup
\small
\setlength{\tabcolsep}{3.5pt}
\renewcommand{\arraystretch}{1.15}

\resizebox{\linewidth}{!}{%
\begin{threeparttable}
\begin{tabular}{@{}llccccc@{}}
\toprule
\textbf{Method}
& \textbf{Metric}
& \textbf{L1}
& \textbf{L2}
& \textbf{L3}
& \textbf{L4}
& \shortstack{\textbf{All hidden}\\\textbf{layers}}
\\
\midrule

\multirow{2}{*}{%
  \shortstack[l]{{\em Normal}\\{\em Alignment}}}
& \textbf{Correct signs}
& $256/256$
& $256/256$
& $255/256$
& $63/64$
& $830/832$
\\
& Vote accuracy $p$ (\%)
& $77.85$
& $72.09$
& $67.85$
& $74.46$
& $72.74$
\\

\midrule

\multirow{3}{*}{%
  \shortstack[l]{{\em eSOE+}\\{\em Alignment}}}
& \textbf{Correct signs}
& $\mathbf{256/256}$
& $\mathbf{256/256}$
& $\mathbf{256/256}$
& $\mathbf{64/64}$
& $\mathbf{832/832}$
\\
& Projected rank
& $255$
& $146$
& $64$
& $63$
& --
\\
& Min confidence
& \tblNA
& $0.73$
& $0.60$
& \tblNA
& $0.60$
\\

\bottomrule
\end{tabular}

\begin{tablenotes}[flushleft]


\item[]
\textbf{Projected rank}:
The rank of the coefficient matrix in Eq.~\eqref{eq:soe-align-stacked-system} after projection and column elimination. In every layer, the reported
rank equals $|\mathbb U|-1$, so the nonzero solution is uniquely
determined up to a scalar multiple.

\item[]
\textbf{Min confidence}:
The minimum {\em Normal Alignment} confidence used to eliminate
inactive-neuron columns. N/A indicates that the unprojected stacked
coefficient matrix already has full column rank, so no column
elimination is required before projection.

\end{tablenotes}
\end{threeparttable}%
}

\endgroup
\end{table}


\begin{table}[!htbp]
\centering
\caption{Layer-wise confidence and ranks of highest-confidence incorrect sign  
predictions produced by {\em Normal Alignment} and
{\em Future Toggle}~\cite{DBLP:conf/eurocrypt/CarliniCHRS25} on the
CIFAR-10 and MNIST models.}
\label{tab:highest-error-confidence-rank}

\begingroup
\fontsize{7.8}{8.5}\selectfont
\setlength{\tabcolsep}{2.2pt}
\renewcommand{\arraystretch}{1.30}

\resizebox{\linewidth}{!}{%
\begin{threeparttable}
\begin{tabular}{@{}l@{\hspace{6pt}}l@{\hspace{8pt}}lcccccccc@{}}
\toprule
\multirow{2}{*}{\raisebox{-0.6ex}{\textbf{Model}}}
& \multirow{2}{*}{\raisebox{-0.6ex}{\textbf{Method}}}
& \multirow{2}{*}{\raisebox{-0.6ex}{\textbf{Metric}}}
& \multicolumn{8}{c}{\textbf{Hidden Layer}} \\
\cmidrule(l){4-11}
& & &
\textbf{L1} & \textbf{L2} & \textbf{L3} & \textbf{L4}
& \textbf{L5} & \textbf{L6} & \textbf{L7} & \textbf{L8} \\
\midrule

\multirow{6}{*}{%
  \raisebox{-3.2ex}{\makecell[l]{\textbf{CIFAR-10}\\
  $192$-$64{\times}8$-$10$}}}
&
\multirow{2}{*}{%
  \raisebox{-0.8ex}{\makecell[l]{{\em Normal Alignment}\\[2pt]
  $n_{\mathrm{attempt}}=200$}}}
& I-Confidence (\%)
& \multirow{2}{*}{$\checkmark$}
& 50.50 & 51.00 & 51.50 & 55.50 & 64.00 & 55.50 & 55.50 \\
& & I-Rank
& & 64/64 & 64/64 & 61/64 & 60/64 & 50/64 & 60/64 & 58/64 \\

\noalign{\vskip 2.5pt}
\cdashline{2-11}[0.4pt/1.5pt]
\noalign{\vskip 2.5pt}

&
\multirow{2}{*}{%
  \raisebox{-0.8ex}{\makecell[l]{{\em Future Toggle}\\[2pt]
  $n_{\mathrm{attempt}}=200$}}}
& I-Confidence (\%)
& \multirow{2}{*}{$\checkmark$}
& 51.05 & 53.01 & 54.55 & 56.89 & 55.70 & 58.72 & 100.00 \\
& & I-Rank
& & 61/64 & 42/64 & 28/64 & 15/64 & 26/64 & 12/64 & 1/64 \\

\noalign{\vskip 2.5pt}
\cdashline{2-11}[0.4pt/1.5pt]
\noalign{\vskip 2.5pt}

&
\multirow{2}{*}{%
  \raisebox{-0.8ex}{\makecell[l]{{\em Future Toggle}\\[2pt]
  $n_{\mathrm{attempt}}=1000$}}}
& I-Confidence (\%)
& \multirow{2}{*}{$\checkmark$}
& \multirow{2}{*}{$\checkmark$}
& \multirow{2}{*}{$\checkmark$}
& 50.62 & 51.90 & 53.36 & 55.58 & 70.00 \\
& & I-Rank
& & & & 60/64 & 44/64 & 45/64 & 20/64 & 13/64 \\

\noalign{\vskip 2pt}
\midrule
\noalign{\vskip 2pt}

\multirow{6}{*}{%
  \raisebox{-3.2ex}{\makecell[l]{\textbf{MNIST}\\
  $64$-$96{\times}3$-$32$-$10$}}}
&
\multirow{2}{*}{%
  \raisebox{-0.8ex}{\makecell[l]{{\em Normal Alignment}\\[2pt]
  $n_{\mathrm{attempt}}=200$}}}
& I-Confidence (\%)
& 55.43 & 62.30 & 60.11 & 56.99
& \multirow{2}{*}{--} & \multirow{2}{*}{--}
& \multirow{2}{*}{--} & \multirow{2}{*}{--} \\
& & I-Rank
& 88/96 & 68/96 & 78/96 & 29/32 & & & & \\

\noalign{\vskip 2.5pt}
\cdashline{2-11}[0.4pt/1.5pt]
\noalign{\vskip 2.5pt}

&
\multirow{2}{*}{%
  \raisebox{-0.8ex}{\makecell[l]{{\em Future Toggle}\\[2pt]
  $n_{\mathrm{attempt}}=200$}}}
& I-Confidence (\%)
& 50.56 & 60.71 & 80.00 & 100.00
& \multirow{2}{*}{--} & \multirow{2}{*}{--}
& \multirow{2}{*}{--} & \multirow{2}{*}{--} \\
& & I-Rank
& 96/96 & 27/96 & 4/96 & 1/32 & & & & \\

\noalign{\vskip 2.5pt}
\cdashline{2-11}[0.4pt/1.5pt]
\noalign{\vskip 2.5pt}

&
\multirow{2}{*}{%
  \raisebox{-0.8ex}{\makecell[l]{{\em Future Toggle}\\[2pt]
  $n_{\mathrm{attempt}}=1000$}}}
& I-Confidence (\%)
& \multirow{2}{*}{$\checkmark$}
& 57.46 & 63.77 & $100.00^{\dagger}$
& \multirow{2}{*}{--} & \multirow{2}{*}{--}
& \multirow{2}{*}{--} & \multirow{2}{*}{--} \\
& & I-Rank
& & 49/96 & 15/96 & 1/32 & & & & \\

\bottomrule
\end{tabular}

\begin{tablenotes}[flushleft]
\item {\bf I-Confidence:} It is the highest confidence among all the \textcolor{red}{\bf i}ncorrect sign predictions. Therefore, the sign-recovery method with lower I-Confidence is better. 


\item {\bf I-Rank:}
Confidence ranks are computed in descending order among all neurons in the corresponding layer; rank $1$ denotes the highest confidence. I-Rank is the confidence rank of the highest-confidence incorrect sign prediction. This metric is critical for confidence-ordered enumeration in complete sign recovery: following \cite{DBLP:conf/nips/FoersterMSH24}, all signs at or below the I-Rank in the confidence ordering, including the sign at the I-Rank itself, must be included in the enumeration. Therefore, a method is better when its highest-confidence error occurs lower in the confidence ordering, i.e., at a larger numerical I-Rank.

\item[$\checkmark$:\,]
No incorrect sign prediction is produced in this layer.

\item[$\dagger$:\,]
Dual-point utilization in the final hidden layer is only $0.21\%$;
there is a neuron with only two valid votes, both of which are incorrect.
\end{tablenotes}
\end{threeparttable}%
}

\endgroup
\end{table}

\begin{sidewaystable}[p]
\centering

\begin{minipage}{0.92\linewidth}
\caption{Layer-wise comparison of {\em Normal Alignment}, {\em eSOE+Alignment}, {\em Future Toggle}, and {\em eSOE+Toggle} on a CIFAR-10 model with architecture \mbox{$192$-$64{\times 8}$-$10$}.}
\label{tab:cifar-layer-comparison}
\end{minipage}

\begingroup
\fontsize{8.8}{10.2}\selectfont
\setlength{\tabcolsep}{2.1pt}
\renewcommand{\arraystretch}{1.10}

\resizebox{0.92\linewidth}{!}{%
\begin{tabular}{@{}llccccccccc@{}}
\toprule
\textbf{Method}
& \textbf{Metric}
& \textbf{L1}
& \textbf{L2}
& \textbf{L3}
& \textbf{L4}
& \textbf{L5}
& \textbf{L6}
& \textbf{L7}
& \textbf{L8}
& \shortstack{\textbf{All hidden}\\\textbf{layers}}
\\
\midrule

\multirow{7}{*}{%
  \textcolor{oursblue}{%
    \shortstack[l]{\textbf{Normal}\\
                   \textbf{Alignment}\\
                $\boldsymbol{n_{\mathrm{attempt}}=200}$}}}
& \textbf{Correct signs}
& 64/64 & 63/64 & 63/64 & 60/64
& 62/64 & 62/64 & 60/64 & 61/64
& 495/512
\\
& Vote accuracy $p$ (\%)
& 75.22 & 69.83 & 69.65 & 68.46
& 70.73 & 70.33 & 70.30 & 71.05
& 70.70
\\
& $\eta_{\rm dual}$ (\%)
& 100.00 & 100.00 & 100.00 & 100.00
& 100.00 & 100.00 & 100.00 & 100.00
& 100.00
\\
& Time
& $2^{1.00}(s)$
& $2^{3.32}(s)+2^{1}(g)$
& $2^{3.17}(s)+2^{1}(g)$
& $2^{3.32}(s)+2^{4}(g)$
& $2^{3.32}(s)+2^{5}(g)$
& $2^{3.32}(s)+2^{15}(g)$
& $2^{3.32}(s)+2^{5}(g)$
& $2^{3.46}(s)+2^{7}(g)$
& $2^{6.17}(s)+2^{15}(g)$
\\
& Queries
& $2^{28.11}$ & $2^{28.11}$ & $2^{28.12}$ & $2^{28.13}$
& $2^{28.13}$ & $2^{28.15}$ & $2^{28.15}$ & $2^{28.15}$
& $2^{31.13}$
\\
& I-Confidence (\%)
& \multirow{2}{*}{$\checkmark$}
& 50.50 & 51.00 & 51.50 & 55.50
& 64.00 & 55.50 & 55.50 & --
\\
& I-Rank
& & 64/64 & 64/64 & 61/64 & 60/64
& 50/64 & 60/64 & 58/64 & --
\\

\noalign{\vskip 2pt}
\cdashline{1-11}[0.4pt/1.5pt]
\noalign{\vskip 2pt}

\multirow{8}{*}{%
  \textcolor{oursblue}{%
    \shortstack[l]{\textbf{eSOE+}\\
    \textbf{Alignment}\\
                $\boldsymbol{n_{\mathrm{attempt}}=200}$}}}
& \textbf{Correct signs}
& 64/64 & 64/64 & 64/64 & 64/64
& 64/64 & 64/64 & 64/64 & 64/64
& 512/512
\\
& Vote accuracy $p$ (\%)
& 75.22 & 69.83 & 69.65 & 68.46
& 70.73 & 70.33 & 70.30 & 71.05
& 70.70
\\
& $\eta_{\rm dual}$ (\%)
& 100.00 & 100.00 & 100.00 & 100.00
& 100.00 & 100.00 & 100.00 & 100.00
& 100.00
\\
& Time
& $2^{6.88}(s)$
& $2^{7.04}(s)$
& $2^{7.18}(s)$
& $2^{7.40}(s)$
& $2^{7.48}(s)$
& $2^{7.82}(s)$
& $2^{5.55}(s)$
& $2^{9.43}(s)$
& $2^{10.73}(s)$
\\
& Queries
& $2^{28.11}$ & $2^{28.11}$ & $2^{28.12}$ & $2^{28.13}$
& $2^{28.13}$ & $2^{28.15}$ & $2^{28.15}$ & $2^{28.15}$
& $2^{31.13}$
\\
& I-Confidence (\%)
& \multirow{2}{*}{$\checkmark$}
& \multirow{2}{*}{$\checkmark$}
& \multirow{2}{*}{$\checkmark$}
& \multirow{2}{*}{$\checkmark$}
& \multirow{2}{*}{$\checkmark$}
& \multirow{2}{*}{$\checkmark$}
& \multirow{2}{*}{$\checkmark$}
& \multirow{2}{*}{$\checkmark$}
& --
\\
& I-Rank
& & & & & & & & & --
\\
& Min confidence
& \tblNA & 0.68 & 0.64 & 0.61
& 0.63 & 0.64 & 0.60 & 0.66
& 0.60
\\

\midrule

\multirow{7}{*}{%
  \shortstack[l]{Future\\ Toggle\\$n_{\mathrm{attempt}}=200$}}
& \textbf{Correct signs}
& 64/64 & 62/64 & 55/64 & 51/64
& 51/64 & 50/64 & 48/64 & 33/64
& 414/512
\\
& Vote accuracy $p$ (\%)
& 100.00 & 56.59 & 54.61 & 53.96
& 53.45 & 54.18 & 53.77 & 56.40
& 62.29
\\
& $\eta_{\rm dual}$ (\%)
& 99.78 & 94.86 & 91.73 & 87.38
& 80.84 & 71.90 & 56.04 & 2.26
& 73.10
\\
& Time
& $2^{7.97}(s)$
& $2^{6.48}(s)+2^{4}(g)$
& $2^{6.66}(s)+2^{23}(g)$
& $2^{6.91}(s)+2^{37}(g)$
& $2^{7.21}(s)+2^{50}(g)$
& $2^{7.5}(s)+2^{39}(g)$
& $2^{8.23}(s)+2^{53}(g)$
& $2^{9.35}(s)+2^{64}(g)$
& $2^{10.85}(s)+2^{64}(g)$
\\
& Queries
& $2^{29.12}$ & $2^{29.33}$ & $2^{29.44}$ & $2^{29.60}$
& $2^{29.81}$ & $2^{30.08}$ & $2^{30.59}$ & $2^{31.57}$
& $2^{33.17}$
\\
& I-Confidence (\%)
& \multirow{2}{*}{$\checkmark$}
& 51.05 & 53.01 & 54.55 & 56.89
& 55.70 & 58.72 & 100.00 & --
\\
& I-Rank
& & 61/64 & 42/64 & 28/64 & 15/64
& 26/64 & 12/64 & 1/64 & --
\\

\noalign{\vskip 2pt}
\cdashline{1-11}[0.4pt/1.5pt]
\noalign{\vskip 2pt}

\multirow{8}{*}{%
  \shortstack[l]{eSOE+\\ Toggle\\$n_{\mathrm{attempt}}=200$}}
& \textbf{Correct signs}
& 64/64 & 64/64 & 64/64 & 51/64
& 47/64 & 49/64 & 50/64 & 36/64
& 425/512
\\
& Vote accuracy $p$ (\%)
& 100.00 & 56.59 & 54.61 & 53.96
& 53.45 & 54.18 & 53.77 & 56.40
& 62.29
\\
& $\eta_{\rm dual}$ (\%)
& 99.78 & 94.86 & 91.73 & 87.38
& 80.84 & 71.90 & 56.04 & 2.26
& 73.10
\\
& Time
& $2^{8.55}(s)$
& $2^{7.77}(s)$
& $2^{7.93}(s)$
& $2^{8.15}(s)+2^{37}(g)$
& $2^{8.31}(s)+2^{50}(g)$
& $2^{8.66}(s)+2^{39}(g)$
& $2^{9.26}(s)+2^{53}(g)$
& $2^{9.56}(s)+2^{64}(g)$
& $2^{11.65}(s)+2^{64}(g)$
\\
& Queries
& $2^{29.12}$ & $2^{29.33}$ & $2^{29.44}$ & $2^{29.60}$
& $2^{29.81}$ & $2^{30.08}$ & $2^{30.59}$ & $2^{31.57}$
& $2^{33.17}$
\\
& I-Confidence (\%)
& \multirow{2}{*}{$\checkmark$}
& \multirow{2}{*}{$\checkmark$}
& \multirow{2}{*}{$\checkmark$}
& 54.55 & 56.89 & 55.70 & 58.72 & 100.00 & --
\\
& I-Rank
& & & & 28/64 & 15/64 & 26/64 & 12/64 & 1/64 & --
\\
& Min confidence
& \tblNA & 0.56 & 0.53 & 0.53
& 0.51 & 0.53 & 0.51 & 0.60
& 0.51
\\

\midrule

\multirow{7}{*}{%
  \shortstack[l]{Future\\ Toggle\\$n_{\mathrm{attempt}}=1000$}}
& \textbf{Correct signs}
& 64/64 & 64/64 & 64/64 & 61/64
& 59/64 & 57/64 & 54/64 & 47/64
& 470/512
\\
& Vote accuracy $p$ (\%)
& 100.00 & 56.42 & 55.34 & 54.24
& 53.59 & 54.42 & 53.40 & 56.89
& 62.45
\\
& $\eta_{\rm dual}$ (\%)
& 99.78 & 95.36 & 91.69 & 87.43
& 80.21 & 71.40 & 55.98 & 2.20
& 73.00
\\
& Time
& $2^{10.28}(s)$
& $2^{8.80}(s)$
& $2^{8.98}(s)$
& $2^{9.24}(s)+2^{5}(g)$
& $2^{9.54}(s)+2^{21}(g)$
& $2^{9.89}(s)+2^{20}(g)$
& $2^{10.5}(s)+2^{45}(g)$
& $2^{11.66}(s)+2^{52}(g)$
& $2^{13.17}(s)+2^{52}(g)$

\\
& Queries
& $2^{31.44}$ & $2^{31.64}$ & $2^{31.76}$ & $2^{31.92}$
& $2^{32.13}$ & $2^{32.41}$ & $2^{32.92}$ & $2^{33.87}$
& $2^{35.49}$
\\

& I-Confidence (\%)
& \multirow{2}{*}{$\checkmark$}
& \multirow{2}{*}{$\checkmark$}
& \multirow{2}{*}{$\checkmark$}
& 50.62 & 51.90 & 53.36 & 55.58 & 70.00 & --
\\
& I-Rank
& & & & 60/64 & 44/64 & 45/64 & 20/64 & 13/64 & --
\\


\noalign{\vskip 2pt}
\cdashline{1-11}[0.4pt/1.5pt]
\noalign{\vskip 2pt}

\multirow{8}{*}{%
  \shortstack[l]{eSOE+\\ Toggle\\$n_{\mathrm{attempt}}=1000$}}
& \textbf{Correct signs}
& 64/64 & 64/64 & 64/64 & 64/64
& 54/64 & 64/64 & 55/64 & 40/64
& 469/512
\\
& Vote accuracy $p$ (\%)
& 100.00 & 56.42 & 55.34 & 54.24
& 53.59 & 54.42 & 53.40 & 56.89
& 62.45
\\
& $\eta_{\rm dual}$ (\%)
& 99.78 & 95.36 & 91.69 & 87.43
& 80.21 & 71.40 & 55.98 & 2.20
& 73.00
\\
& Time
& $2^{10.41}(s)$
& $2^{9.17}(s)$
& $2^{9.35}(s)$
& $2^{9.60}(s)$
& $2^{9.86}(s)+2^{21}(g)$
& $2^{10.20}(s)$
& $2^{10.83}(s)+2^{45}(g)$
& $2^{11.98}(s)+2^{52}(g)$
& $2^{13.47}(s)+2^{52}(g)$

\\
& Queries
& $2^{31.44}$ & $2^{31.64}$ & $2^{31.76}$ & $2^{31.92}$
& $2^{32.13}$ & $2^{32.41}$ & $2^{32.92}$ & $2^{33.87}$
& $2^{35.49}$
\\
& I-Confidence (\%)
& \multirow{2}{*}{$\checkmark$}
& \multirow{2}{*}{$\checkmark$}
& \multirow{2}{*}{$\checkmark$}
& \multirow{2}{*}{$\checkmark$}
& 51.90
& \multirow{2}{*}{$\checkmark$}
& 55.58 & 70.00 & --
\\
& I-Rank
& & & & & 44/64 & & 20/64 & 13/64 & --
\\
& Min confidence
& \tblNA & 0.56 & 0.55 & 0.53
& 0.52 & 0.54 & 0.52 & 0.58
& 0.52
\\

\bottomrule
\end{tabular}%
}
\endgroup







\par\smallskip
\begin{minipage}{0.92\linewidth}
\scriptsize
\begin{tabularx}{\linewidth}{@{}l@{\hspace{0.5em}}X@{}}

\textbf{I-Confidence/I-Rank}: &
These metrics are defined in Table~\ref{tab:highest-error-confidence-rank}.
A $\checkmark$ indicates that no incorrect sign prediction is produced in the corresponding layer. If a method combined with {\em eSOE} does not recover all signs in a layer, we report the I-Confidence and I-Rank
of its underlying statistical method, {\em Normal Alignment} or {\em Future Toggle}, because the fallback exhaustive search uses that method's confidence ordering.
\\

\textbf{Min confidence}: &
The minimum confidence used by {\em eSOE} to eliminate
inactive-neuron columns. N/A indicates that the unprojected stacked
coefficient matrix already has full column rank, so no column
elimination is required before projection.
\\

\textbf{Time}: &
$2^x(\mathrm{s})$ denotes the method runtime in seconds, whereas
$2^y(\mathrm{g})$ denotes the estimated number of candidate sign
assignments required by confidence-ordered exhaustive search for
complete recovery. The guessing term is omitted when all signs are
correctly recovered.
\\

\end{tabularx}
\end{minipage}

\end{sidewaystable}

\begin{sidewaystable}
\centering

\begin{minipage}{0.72\linewidth}
\caption{Layer-wise comparison of {\em Normal Alignment}, {\em eSOE+Alignment}, {\em Future Toggle}, and {\em eSOE+Toggle} on an MNIST model with architecture \mbox{$64$-$96{\times 3}$-$32$-$10$}.}
\label{tab:mnist-layer-comparison}
\end{minipage}

\begingroup
\fontsize{8.6}{9.4}\selectfont
\setlength{\tabcolsep}{2.8pt}
\renewcommand{\arraystretch}{1.02}

\resizebox{0.72\linewidth}{!}{%
\begin{tabular}{@{}llccccc@{}}
\toprule
\textbf{Method}
& \textbf{Metric}
& \textbf{L1}
& \textbf{L2}
& \textbf{L3}
& \textbf{L4}
& \shortstack{\textbf{All hidden}\\\textbf{layers}}
\\
\midrule

\multirow{7}{*}{%
  \textcolor{oursblue}{%
    \shortstack[l]{\textbf{Normal}\\
                   \textbf{Alignment}\\
                $\boldsymbol{n_{\mathrm{attempt}}=200}$}}}
& \textbf{Correct signs}
& 91/96 & 93/96 & 93/96 & 30/32 & 307/320
\\
& Vote accuracy $p$ (\%)
& 70.95 & 64.98 & 69.10 & 72.33 & 68.74
\\
& $\eta_{\rm dual}$ (\%)
& 88.24 & 89.71 & 91.05 & 90.53 & 89.75
\\
& Time
& $2^{8.94}(s)+2^{9}(g)$
& $2^{9.03}(s)+2^{29}(g)$
& $2^{8.84}(s)+2^{19}(g)$
& $2^{7.25}(s)+2^{4}(g)$
& $2^{10.67}(s)+2^{29}(g)$
\\
& Queries
& $2^{26.96}$ & $2^{26.97}$ & $2^{26.98}$
& $2^{25.40}$ & $2^{28.71}$
\\
& I-Confidence (\%)
& 55.43 & 62.30 & 60.11 & 56.99 & --
\\
& I-Rank
& 88/96 & 68/96 & 78/96 & 29/32 & --
\\

\noalign{\vskip 1pt}
\cdashline{1-7}[0.4pt/1.5pt]
\noalign{\vskip 1pt}

\multirow{8}{*}{%
  \textcolor{oursblue}{%
    \shortstack[l]{\textbf{eSOE+}\\
    \textbf{Alignment}\\
                $\boldsymbol{n_{\mathrm{attempt}}=200}$}}}
& \textbf{Correct signs}
& 96/96 & 96/96 & 96/96 & 32/32 & 320/320
\\
& Vote accuracy $p$ (\%)
& 70.95 & 64.98 & 69.10 & 72.33 & 68.74
\\
& $\eta_{\rm dual}$ (\%)
& 88.24 & 89.71 & 91.05 & 90.53 & 89.75
\\
& Time
& $2^{9.16}(s)$
& $2^{9.29}(s)$
& $2^{9.23}(s)$
& $2^{9.10}(s)$
& $2^{11.20}(s)$
\\
& Queries
& $2^{26.96}$ & $2^{26.98}$ & $2^{26.98}$
& $2^{25.40}$ & $2^{28.71}$
\\
& I-Confidence (\%)
& \multirow{2}{*}{$\checkmark$}
& \multirow{2}{*}{$\checkmark$}
& \multirow{2}{*}{$\checkmark$}
& \multirow{2}{*}{$\checkmark$}
& --
\\
& I-Rank
& & & & & --
\\
& Min confidence
& 0.74 & 0.63 & 0.61 & \tblNA & 0.61
\\

\midrule

\multirow{7}{*}{%
  \shortstack[l]{Future\\ Toggle\\$n_{\mathrm{attempt}}=200$}}
& \textbf{Correct signs}
& 95/96 & 86/96 & 66/96 & 8/32 & 255/320
\\
& Vote accuracy $p$ (\%)
& 63.67 & 57.37 & 56.85 & 75.00 & 60.89
\\
& $\eta_{\rm dual}$ (\%)
& 92.32 & 60.42 & 11.26 & 0.25 & 49.23
\\
& Time
& $2^{13.30}(s)+2^{1}(g)$
& $2^{13.35}(s)+2^{70}(g)$
& $2^{13.41}(s)+2^{93}(g)$
& $2^{11.93}(s)+2^{32}(g)$
& $2^{15.11}(s)+2^{93}(g)$
\\
& Queries
& $2^{26.63}$ & $2^{26.68}$ & $2^{26.77}$
& $2^{25.31}$ & $2^{28.45}$
\\
& I-Confidence (\%)
& 50.56 & 60.71 & 80.00 & 100.00 & --
\\
& I-Rank
& 96/96 & 27/96 & 4/96 & 1/32 & --
\\

\noalign{\vskip 1pt}
\cdashline{1-7}[0.4pt/1.5pt]
\noalign{\vskip 1pt}

\multirow{8}{*}{%
  \shortstack[l]{eSOE+\\Toggle\\$n_{\mathrm{attempt}}=200$}}
& \textbf{Correct signs}
& 96/96 & 83/96 & 67/96 & 32/32 & 278/320
\\
& Vote accuracy $p$ (\%)
& 63.67 & 57.37 & 56.85 & 75.00 & 60.89
\\
& $\eta_{\rm dual}$ (\%)
& 92.32 & 60.42 & 11.26 & 0.25 & 49.23
\\
& Time
& $2^{13.31}(s)$
& $2^{13.36}(s)+2^{70}(g)$
& $2^{13.41}(s)+2^{93}(g)$
& $2^{12.08}(s)$
& $2^{15.13}(s)+2^{93}(g)$
\\
& Queries
& $2^{26.63}$ & $2^{26.69}$ & $2^{26.77}$
& $2^{25.31}$ & $2^{28.45}$
\\
& I-Confidence (\%)
& \multirow{2}{*}{$\checkmark$}
& 60.71 & 80.00
& \multirow{2}{*}{$\checkmark$}
& --
\\
& I-Rank
& & 27/96 & 4/96 & & --
\\
& Min confidence
& 0.63 & 0.56 & 0.52 & \tblNA & 0.52
\\

\midrule

\multirow{7}{*}{%
  \shortstack[l]{Future\\ Toggle\\$n_{\mathrm{attempt}}=1000$}}
& \textbf{Correct signs}
& 96/96 & 82/96 & 80/96 & 19/32 & 277/320
\\
& Vote accuracy $p$ (\%)
& 63.05 & 56.72 & 57.31 & 71.21 & 60.32
\\
& $\eta_{\rm dual}$ (\%)
& 91.93 & 60.37 & 11.43 & 0.21 & 49.14
\\
& Time
& $2^{15.57}(s)$
& $2^{15.63}(s)+2^{48}(g)$
& $2^{15.66}(s)+2^{82}(g)$
& $2^{14.21}(s)+2^{32}(g)$
& $2^{17.38}(s)+2^{82}(g)$
\\
& Queries
& $2^{28.95}$ & $2^{29.00}$ & $2^{29.09}$
& $2^{27.63}$ & $2^{30.77}$
\\
& I-Confidence (\%)
& \multirow{2}{*}{$\checkmark$}
& 57.46 & 63.77 & $100.00$ & --
\\
& I-Rank
& & 49/96 & 15/96 & 1/32 & --
\\

\noalign{\vskip 1pt}
\cdashline{1-7}[0.4pt/1.5pt]
\noalign{\vskip 1pt}

\multirow{8}{*}{%
  \shortstack[l]{eSOE+\\ Toggle\\$n_{\mathrm{attempt}}=1000$}}
& \textbf{Correct signs}
& 96/96 & 96/96 & 72/96 & 32/32 & 296/320
\\
& Vote accuracy $p$ (\%)
& 63.05 & 56.72 & 57.31 & 71.21 & 60.32
\\
& $\eta_{\rm dual}$ (\%)
& 91.93 & 60.37 & 11.43 & 0.21 & 49.14
\\
& Time
& $2^{15.57}(s)$
& $2^{15.64}(s)$
& $2^{15.67}(s)+2^{82}(g)$
& $2^{14.25}(s)$
& $2^{17.39}(s)+2^{82}(g)$
\\
& Queries
& $2^{28.95}$ & $2^{29.01}$ & $2^{29.09}$
& $2^{27.63}$ & $2^{30.77}$
\\
& I-Confidence (\%)
& \multirow{2}{*}{$\checkmark$}
& \multirow{2}{*}{$\checkmark$}
& 63.77
& \multirow{2}{*}{$\checkmark$}
& --
\\
& I-Rank
& & & 15/96 & & --
\\
& Min confidence
& 0.64 & 0.56 & 0.52 & \tblNA & 0.52
\\

\bottomrule
\end{tabular}%
}
\endgroup

\par\smallskip
\begin{minipage}{0.72\linewidth}
\scriptsize
\begin{tabularx}{\linewidth}{@{}l@{\hspace{0.5em}}X@{}}

\textbf{I-Confidence/I-Rank}: &
These metrics are defined in
Table~\ref{tab:highest-error-confidence-rank}.
A $\checkmark$ indicates that no incorrect sign prediction is produced
in the corresponding layer. If a method combined with {\em eSOE} does
not recover all signs in a layer, we report the I-Confidence and I-Rank
of its underlying statistical method, {\em Normal Alignment} or
{\em Future Toggle}, because the fallback exhaustive search uses that
method's confidence ordering.
\\

\textbf{Min confidence}: &
The minimum confidence used by {\em eSOE} to eliminate
inactive-neuron columns. N/A indicates that the unprojected stacked
coefficient matrix already has full column rank, so no column
elimination is required before projection.
\\

\textbf{Time}: &
$2^x(\mathrm{s})$ denotes the method runtime in seconds, whereas
$2^y(\mathrm{g})$ denotes the estimated number of candidate sign
assignments required by confidence-ordered exhaustive search for
complete recovery. The guessing term is omitted when all signs are
correctly recovered.
\\

\end{tabularx}
\end{minipage}
\end{sidewaystable}